\documentclass[11pt,showkeys]{article}
\usepackage[utf8x]{inputenc}
\usepackage{amssymb,amsfonts,amssymb,bbm}
\usepackage{tikz}
\usepackage{hyperref}

\def\Rset{\mathbb{R}} 
\def\Sset{\mathbb{S}} 
\def\un{\mathbbm{1}} 
\def\Real{\Re e} 
\def\;{\, ; \,} 
\def\opA{\mathfrak{A}} 
\def\opE{\mathfrak{E}} 
\def\opT{\mathfrak{T}} 
\def\Beta{\mathfrak{B}} 
\def\G{\mathfrak{G}} 
\def\Aset{\mathcal{A}} 
\def\Bset{\mathcal{B}} 
\def\Dset{\mathcal{D}} 
\def\Lset{\mathcal{L}} 
\def\Rad{\mathcal{R}} 
\def\Mad{\mathcal{M}} 
\def\Hyp{\mathcal{Y}} 
\def\Lag{\mathcal{L}} 
\def\Geg{\mathcal{G}} 
\def\vr{\vec{r}} 
\def\vk{\vec{k}} 
\def\tr{\tilde{r}} 
\def\tk{\tilde{k}} 
\def\hyd{{(\mathrm{h})}} 
\def\osc{{(\mathrm{o})}} 
\def\d{d} 
\newtheorem{Definition}{Definition}
\newtheorem{Proposition}{Proposition}
\newtheorem{proof}{proof}
\newtheorem{Property}{Property}

\title{On Generalized Stam Inequalities and Fisher--R\'enyi Complexity Measures}
\author{\large Steeve Zozor\footnote{steeve.zozor@gipsa-lab.grenoble-inp.fr},\vspace{0.3cm}\\\footnotesize GIPSA-Lab, Universit\'e Grenoble Alpes,\\\footnotesize  11 rue des Math\'ematiques, Grenoble 38420, France\vspace{0.3cm}\\\large  David Puertas-Centeno\footnote{puertascenteno@gmail.com}\, and Jes\'us S. Dehesa\footnote{dehesa@ugr.es}\vspace{0.3cm}\\\footnotesize Departamento  de   F\'isica  At\'omica,  Molecular   y  Nuclear,\\\footnotesize   Universidad de Granada, Granada 18071, Spain and\\\footnotesize  Instituto  Carlos   I  de  F\'isica  Te\'orica  y  Computacional,\\\footnotesize   Universidad  de   Granada,  Granada  18071,   Spain}

\date{12 September 2017}
\begin{document}
\maketitle
\small{Information-theoretic inequalities play a fundamental role in numerous
  scientific  and  technological   areas  (e.g.,  estimation  and  communication
  theories, signal and information  processing, quantum physics, \ldots) as they
  generally express the impossibility to have a complete description of a system
  via a finite number of information measures.  In particular, they gave rise to
  the  design of various  quantifiers (statistical  complexity measures)  of the
  internal  complexity of a  (quantum) system.   In this  paper, we  introduce a
  three-parametric           Fisher--R\'enyi          complexity,          named
  $(p,\beta,\lambda)$-Fisher--R\'enyi complexity, based  on both a two-parametic
  extension of the Fisher information and the R\'enyi entropies of a probability
  density function $\rho$ characteristic of the system.  This complexity measure
  quantifies the combined balance of  the spreading and the gradient contents of
  $\rho$, and  has the  three main properties  of a statistical  complexity: the
  invariance  under translation  and  scaling transformations,  and a  universal
  bounding  from  below.   The  latter   is  proved  by  generalizing  the  Stam
  inequality, which lowerbounds the product of the Shannon entropy power and the
  Fisher information  of a probability  density function.  An extension  of this
  inequality was  already proposed by Bercher  and Lutwak, a  particular case of
  the general one, where the  three parameters are linked, allowing to determine
  the  sharp lower  bound and  the associated  probability density  with minimal
  complexity.  Using the notion  of differential-escort deformation, we are able
  to determine  the sharp bound  of the complexity  measure even when  the three
  parameters  are decoupled  (in a  certain range).   We determine  as  well the
  distribution that  saturates the inequality:  the $(p,\beta,\lambda)$-Gaussian
  distribution, which  involves an  inverse incomplete beta  function.  Finally,
  the complexity measure is  calculated for various quantum-mechanical states of
  the  harmonic and hydrogenic  systems, which  are the  two main  prototypes of
  physical systems subject to a central potential.}

\vspace{0.3cm}
\textit{Keywords: }{$(p,\beta,\lambda)$-Fisher--R\'enyi  complexity;  extended  sharp  Stam
  inequality;   $(p,\beta,\lambda)$-Gaussian   distributions;   application   to
  $d$-dimensional central potential quantum systems}

\section{Introduction}
\normalsize
The  definition of  complexity measures  to  quantify the  internal disorder  of
physical  systems is  an important  and challenging  task in  science, basically
because  of the  many facets  of the  notion of  disorder~\cite{Sen11, LopMan95,
  Lop05,   ChaMou05,   SenPan07,   BiaRud10,   DehLop10:12,   Hua13,   EbeMol00,
  AngAnt08:04, RosOsp06,  TorSan17}.  It  seems clear that  a unique  measure is
unable  to  capture  the  essence  of  such  a  vague  notion.   In  the  scalar
continuous-state  context we consider  in this  paper, many  complexity measures
based on the probability distribution  describing a system have been proposed in
the literature, attempting to  capture simultaneously the spreading (global) and
the  oscillatory  (local)   behaviors  of  such  a  distribution~\cite{AngRom00,
  DehLop10:07,  LopEsq10,  RomSan06,  SanZoz11,  ZozPor11,  MarPla00,  PorPla96,
  MasPan98,  GueSan11,  DehSan06,  AntAng09,  AngAnt08,  RomDeh04,  AngAnt08:04,
  TorSan17,  PueTor17,  SobPue17,  PueTor17:preprint}.   They mostly  depend  on
entropy-like  quantities such  as the  Shannon entropy~\cite{Sha48},  the Fisher
information~\cite{Fis22} and their  generalizations.  The measures of complexity
of a probability density $\rho$ proposed up until now, say $C[\rho]$, making use
of   two  information-theoretic  properties,   share  several   properties  (see
e.g.,~\cite{RudTor16}),  such as  e.g., the  invariance by  translation or  by a
scaling  factor  (i.e.,  for  any  $x_0   \in  \Rset$  and  $\sigma  >  0$,  for
$\widetilde{\rho}(x) = \frac{1}{\sigma} \rho\left(\frac{x-x_0}{\sigma} \right),$
they satisfy  $C[\widetilde{\rho}] = C[\rho]$).  For instance,  the disorder may
be  invariant  from  a move  of  a  (referential  independent) center  of  mass.
Moreover,  all the  proposed measures  are also  lowerbounded, which  means that
there exists in  a certain sense a distribution of  minimal complexity, which is
the probability density that reaches the lower~bound.

In this paper,  we generalize the complexity measures  of global-local character
published in  the literature  (see e.g.,~\cite{DehSan06,  AngAnt08:04, AntAng09,
  RomDeh04, TorSan17, PueTor17, PueTor17:preprint})  to grasp both the spreading
and the  fluctuations of a probability  density $\rho$ by the  introduction of a
three-parametric   Fisher--R\'enyi  complexity,   which  involves   the  R\'enyi
entropy~\cite{Ren61} and generalized Fisher information~\cite{LutYan05, LutLv12,
  Ber12:06_1}.   The  products of  these  two generalized  information-theoretic
tools, which are translation and  scaling invariant as well as lowerbounded, can
be used as generalized complexity measures of $\rho$.

Historically, the first inequality involving  the Shannon entropy and the Fisher
information was proved by Stam~\cite{Sta59} under the form
\begin{equation}
F[\rho] N[\rho] \ge 2 \pi e,
\label{eq:Stam}
\end{equation}
where $F$ and  $N$ are, respectively, the (nonparametric)  Fisher information of
$\rho$,
\begin{equation}
F[\rho] = \int_\Rset \left( \frac{d}{dx} \log[\rho(x)] \right)^2 \, \rho(x) \,
dx
\end{equation}
and the  Shannon entropy power  of $\rho$, i.e.,  an exponential of  the Shannon
entropy $H$,
\begin{equation}
N[\rho] = \exp\left( 2 H[\rho] \right) \qquad \mbox{where} \qquad H[\rho] = -
\int_\Rset \rho(x) \, \log[\rho(x)] \, dx.
\end{equation}
In fact, the  Fisher information concerns a density  parametrized by a parameter
$\theta$ and  the derivative is vs  $\theta$. When this parameter  is a position
parameter, this  leads to the nonparametric Fisher  information.  Concerning the
entropy power, more rigorously, a factor $\frac{1}{2 \pi e}$ affects $N$ and the
bound  in the  Stam  inequality is  then  unity.  This  factor  does not  change
anything for our purpose, hence, for  sake of simplicity, we omit it.  The lower
bound in  Inequality~ \ref{eq:Stam} is  achieved for the  Gaussian distribution
$\rho(x)  \propto  \exp\left(-\frac12 x^2\right)$  up  to  a  translation and  a
scaling factor (where $\propto$ means ``proportional to'').  In other words, the
so-called {\em Fisher--Shannon complexity} $C[\rho] = F[\rho] N[\rho]$, which is
translation  and scale  invariant, is  always higher  than $2  \pi e$  (and thus
cannot  be zero)  and the  distribution of  lowest complexity  is  the Gaussian,
exhibiting (also) through this measure its fundamental aspect. The proof of this
inequality lies in  the entropy power inequality and on  the de Bruijn identity,
two  information theoretic  inequalities,  both being  reached  in the  Gaussian
context~\cite{Sta59, CovTho06}.
Although  introduced   respectively  in  the  estimation   context  through  the
Cram\'er--Rao  bound~\cite{Kay93, LehCas98, Fis22}  and in  communication theory
through the  coding theorem of Shannon~\cite{Sha48,  CovTho06}, these quantities
found applications in  physics as previously mentioned (and  also in the earlier
papers~\cite{Bou58, Lei59} and that of  Stam).  In particular, the analysis of a
signal with  these measures was  proposed by Vignat  and Bercher~\cite{VigBer03}
and the Fisher--Shannon complexity $C[\rho] = F[\rho] N[\rho]$ is widely applied
in atomic  physics or quantum mechanics  for instance~\cite{RomDeh04,
  AngAnt08, SanLop08, DehLop09, LopSan12, Man12}.

Recently, the Stam  inequality was extended by substituting  the Shannon entropy
by the  R\'enyi entropies (a family  of entropies characterizing  by a parameter
playing  a  role  of  focus~\cite{Ren61}),  and  the  Fisher  information  by  a
generalized   two-parametric  family  of   the  Fisher   information  introduced
by~\cite{LutYan05, Ber12:06_1, LutLv12}.  As we will see later on, this extended
inequality involves, however, two free  parameters because one of the two Fisher
parameters is  linked to the R\'enyi one.   This constraint is imposed  so as to
determine the sharp bound of the  inequality and the minimizers in the framework
of  the (stretched)  Tsallis distributions~\cite{GelTsa04,  Tsa09}.   Thus, this
extended inequality allows  to define again a complexity  measure, based on this
generalized Fisher information and the R\'enyi entropy power~\cite{PueTor17}.

In this  paper, we study  the full three-parametric  Fisher--R\'enyi complexity,
disconnecting the two parameters tuning  the extended Fisher information and the
parameter tuning the R\'enyi entropy.  Like Bercher, we use an approach based on
the  Gagliardo--Nirenberg inequality.   This inequality  allows for  proving the
existence of a lower bound of  the complexity when the parameters are decoupled,
in a certain range.  The minimizers are thus implicitly known as a solution of a
nonlinear  equation  (or  through  a  complicated  series  of  integrations  and
inversion of nonlinear functions).  Moreover,  the sharp bound of the associated
extended  Stam inequality  is explicitly  known, once  the minimizers  have been
determined.  We propose here an indirect  approach allowing (i) to extend a step
further the domain  where the Stam inequality holds (or  where the complexity is
non trivially  lower-bounded); (ii) to determine explicitly  the minimizers; and
(iii) to find the sharp bound, regardless of the knowledge of the minimizers.

The structure of the paper is the following. In Section~\ref{sec:Complexity}, we
introduce   both  the   $\lambda$-dependent  R\'enyi   entropy  power   and  the
$(p,\beta)$-Fisher information, so generalizing the usual (i.e., translationally
invariant) Fisher information.   Then, we propose a complexity  measure based on
these   two  information  quantities,   the  $(p,\beta,\lambda)$-Fisher--R\'enyi
complexity,  and we study  its fundamental  properties regarding  the invariance
under  translation and  scaling transformations  and, above  all,  the universal
bounding from  below.  In  particular, we  come back briefly  to the  results of
Lutwak~\cite{LutYan05} or of  Bercher~\cite{Ber12:06_1} concerning the sharpness
of the bound  and the minimizers, derived only when  the three parameters belong
to a two-dimensional manifold, finding that our results remain indeed valid in a
domain slightly wider than  theirs.  In Section~\ref{sec:ExtendedStam}, the core
of  the paper,  we  come  back to  the  lower bound  (or  to  the extended  Stam
inequality)  dealing with  a wide  three-dimensional domain.   In  this extended
domain,  which includes  that of  the previous  section, we  are able  to derive
explicitly the minimizers and the sharp lower bound, regardless the knowledge of
the  minimizers.   In  order  to  do  this, we  introduce  a  special  nonlinear
stretching  of   the  state,   leading  to  the   so-called  differential-escort
distribution~\cite{PueRud17}.  This  geometrical deformation allows  us to start
from the Bercher--Lutwak  inequality and to introduce a  supplementary degree of
freedom so as to decouple the  parameters (in a certain range). This approach is
the key point for the determination  of the extended domain where the complexity
is bounded from  below (the generalized Stam inequality).   Moreover, we provide
an  explicit  expression  for  the  densities which  minimize  this  complexity,
expression    involving   the    inverse   incomplete    beta    function.    In
Section~\ref{sec:Applications}, we  apply the previous results  to some relevant
multidimensional  physical  systems  subject   to  a  central  potential,  whose
quantum-mechanically allowed  stationary states are described  by wave functions
that factorize  into a  potential-dependent radial part  and a  common spherical
part.  Focusing on the radial part, we calculate the three-parametric complexity
of the  two main prototypes  of $\d$-dimensional physical systems,  the harmonic
(i.e., oscillator-like)  and hydrogenic systems,  for various quantum-mechanical
states and  dimensionalities.  Finally,  three appendices containing  details of
the proofs of various propositions of the paper are reported.


\section{$(p,\beta,\lambda)$-Fisher--R\'enyi  Complexity and  the  Extended Stam
  Inequality}
\label{sec:Complexity}

In this section, we firstly review the extension of the Stam inequality based on
the efforts  of Lutwak et al. and  Bercher~\cite{LutYan05, LutLv12, Ber12:06_1},
or more generally, based on that  of Agueh~\cite{Agu06, Agu08}.  To this aim, we
introduce a three-parametric Fisher--R\'enyi complexity, showing its scaling and
translation invariance and  non-trivial bounding from below.  We  then come back
to the  results of Lutwak or  Bercher concerning the determination  of the sharp
bound and the minimizers of its associated complexity, where a constraint on the
parameters was  imposed.  Indeed,  the constraint they  imposed can  be slightly
relaxed, as we will see in this section.


\subsection{R\'enyi  Entropy, Extended  Fisher  Information and  R\'enyi--Fisher
  Complexity}

Let  us  begin  with  the  definitions of  the  following  information-theoretic
quantities  of  the  probability  density  $\rho$:  the  R\'enyi  entropy  power
$N_\lambda[\rho]$,  the $(p,\beta)$-Fisher information  $F_{p,\beta}[\rho]$, and
the $(p,\beta,\lambda)$-Fisher--R\'enyi complexity $C_{p,\beta,\lambda}[\rho]$.
\begin{Definition}[R\'enyi entropy power~\cite{Ren61}]\label{def:RenyiPower}
  Let $\lambda \in  \Rset_+^*$.  Provided that the integral  exists, the R\'enyi
  entropy power of  index $\lambda$ of a probability  density function $\rho$ is
  given by
  \begin{equation}
  N_\lambda[\rho] = \exp\left( 2 H_\lambda[\rho] \right) \qquad
  \mbox{where} \qquad H_\lambda[\rho] = \frac{1}{1-\lambda}
  \log\int_\Rset[\rho(x)]^\lambda \, dx,
  \label{eq:RenyiPower}
  \end{equation}
  where  the limiting  case  $\lambda \to  1$  gives the  Shannon entropy  power
  $\displaystyle   N[\rho]    =   N_1[\rho]   \equiv    \lim_{\lambda   \to   1}
  N_\lambda[\rho]$.
\end{Definition}
The  entropy  $H_\lambda$  was   introduced  by  R\'enyi  in~\cite{Ren61}  as  a
generalization of the Shannon entropy.  In this expression, through the exponent
$\lambda$  applied  to  the distribution,  more  weight  is  given to  the  tail
($\lambda    <   1$)   or    to   the    head   ($\lambda    >   1$)    of   the
distribution~\cite{LutYan05, CosHer03,  JohVig07, NanMai07}. This  measure found
many    applications    in    numerous    fields   such    as    e.g.,    signal
processing~\cite{PanDit51,  Llo82,   GerGra92,  Cam65,  Hum81,   Bae06,  Ber09},
information theory to  reformulate the entropy power inequality~\cite{BobChi15},
statistical     inference~\cite{Par06},    multifractal    analysis~\cite{Har01,
  JizAri04}, chaotic systems~\cite{BecSch93}, or  in physics as mentioned in the
introduction (see ref.   above).  For instance, the R\'enyi  entropies were used
to reformulate the  Heisenberg uncertainty principle (see~\cite{Bia06, ZozVig07,
  ZozVig07:09,  ZozPor08,  ZozBos14}   or~\cite{JizDun15,  JizMa16}  where  this
formulation also appears and is applied in quantum physics).

Whereas  the power  applied to  the probability  density $\rho$  in  the R\'enyi
entropy aims at  making a focus on  heads or tails of the  distribution, one may
wish to act similarly dealing with  the Fisher information.  In this case, since
both the density and its derivative  are involved, one may wish to stress either
some  parts of  the distribution,  or  some of  its variations  (small or  large
fluctuations).   Thus,  two  different  power  parameters  for  $\rho$  and  its
derivative, respectively,  can be considered  leading with our notations  to the
following definition of the bi-parametric Fisher information.
\begin{Definition}[$(p,\beta)$-Fisher    information~\cite{LutYan05,    LutLv12,
    Ber12:06_1}]\label{def:biparametricFisher}
  For   any  $p  \in   (1,\infty  )$   and  any   $\beta  \in   \Rset_+^*$,  the
  $(p,\beta)$-Fisher information of a continuously differentiable density $\rho$
  is defined by
  \begin{equation}
  F_{p,\beta}[\rho] = \left( \int_\Rset \Big| [\rho(x)]^{\beta-1} \, \frac{d}{dx}
  \log[\rho(x)] \Big|^p \, \rho(x) \, dx \right)^\frac{2}{p \beta},
  \label{eq:biparametricFisher}
  \end{equation}
  provided  that this integral  exists. When  $\rho$ is  strictly positive  on a
  bounded support, the integration is to be understood over this support, but it
  must be differentiable on the closure of this support.
\end{Definition}
It is  straightforward to  see that $F_{2,1}$  is the usual  Fisher information.
When     it     exists,     $\displaystyle     \lim_{p     \to     +     \infty}
[F_{p,\beta}]^{\frac{\beta}{2}}$   is   the   essential  supremum   of   $\left|
  \rho^{\beta-1} \frac{d}{dx} \log[\rho] \right|$.  Conversely, $\frac{1}{\beta}
[F_{1,\beta}]^{\frac{\beta}{2}}$ is the total variation of $\rho^\beta$.  For $p
= 2$, this extended Fisher information is closely related to the $\alpha$-Fisher
information  introduced  by  Hammad  in   1978  when  dealing  with  a  position
parameter~\cite{Ham78}.   Note also  that a variety of generalized  Fisher
information was applied  especially in  non-extensive physics~\cite{PenPla98,
  ChiPen00, CasChi02, PenPla07}.

From the R\'enyi entropy power and the $(p,\beta)$-Fisher information, we define
a  $(p,\beta,\lambda)$-Fisher--R\'enyi  complexity   by  the  product  of  these
quantities, up to a given power.
\begin{Definition}[$(p,\beta,\lambda)$-Fisher--R\'enyi
  complexity]\label{def:triparametricComplexity}
  We define the  $(p,\beta,\lambda)$-Fisher--R\'enyi complexity of a probability
  density $\rho$ by
  \begin{equation}
  C_{p,\beta,\lambda}[\rho] = \Big( F_{p,\beta}[\rho] \, N_\lambda[\rho]
  \Big)^\beta,
 \label{eq:triparametricComplexity}
  \end{equation}
  provided that the involved  quantities exist.
\end{Definition}
We choose to elevate the product  of the entropy power and Fisher information to
the power  $\beta > 0$ for  simplification reasons.  Indeed, it  does not change
the spirit of this measure of  complexity, whereas it allows to express symmetry
properties in a more elegant manner, as we will see later on.

This quantity has the minimal  properties expected for a complexity measure (see
e.g.,~\cite{RudTor16}), as stated in the next subsection.


\subsection{Shift  and  Scale Invariance, Bounding  from below and  Minimizing
  Distributions}

The first property of the proposed complexity $C_{p,\beta,\lambda}[\rho]$ is the
invariance under the basic translation and scaling transformations.
\begin{Proposition}\label{prop:Invariance}
  The $(p,\beta,\lambda)$-Fisher--R\'enyi complexity  of the probability density
  $\rho$ is invariant  under any translation $x_0 \in  \Rset$ and scaling factor
  $\sigma   >  0$   applied  to   $\rho$;  i.e.,   for   $\widetilde{\rho}(x)  =
  \frac{1}{\sigma}    \,     \rho\left(    \frac{x-x_0}{\sigma}    \right),    \
  C_{p,\beta,\lambda}[\widetilde{\rho}] = C_{p,\beta,\lambda}[\rho]$.
\end{Proposition}
\begin{proof}
  This  is a  direct consequence  of  a change  of variables  in the  integrals,
  showing   that   $N_\lambda[\widetilde{\rho}]   =  \sigma^2   N_\lambda[\rho]$
  (justifying  the  term   of  entropy  power)  for  any   $\lambda$,  and  that
  $F_{p,\beta}[\widetilde{\rho}]  =   \sigma^{-2}  F_{p,\beta}[\rho]$,  whatever
  $(p,\beta)$.
\end{proof}

From now,  due to these properties,  all the definitions  related to probability
density functions will be given up to a translation and scaling factor. In other
words, when evoking  a density $\rho$, except when specified,  we will deal with
the family  $\frac{1}{\sigma} \, \rho\left(\frac{x-x_0}{\sigma}\right)$  for any
$x_0 \in \Rset$ and $\sigma > 0$.

More important, the complexity has  a universal, non-trivial bounding from below
so that  the distribution corresponding to  this minimal complexity  can thus be
viewed as the less complex one.

\begin{Proposition}[Extended Stam inequality]\label{prop:StamDp}
  For any $p > 1$,
  \begin{equation}
  (\beta,\lambda) \in \Dset_p = \left\{ (\beta,\lambda) \in \Rset_+^{* \, 2}:
  \quad \beta \in \left( \frac{1}{p^*} \; \frac{1}{p^*} + \min(1,\lambda)
  \right] \right\},
  \label{eq:Dset}
  \end{equation}
  with  $p^*  = \frac{p}{p-1}$  the  Holder conjugate  of  $p$,  their exists  a
  universal  optimal positive constant  $K_{p,\beta,\lambda}$, that  bounds from
  below  the  $(p,\beta,\lambda)$-Fisher--R\'enyi   complexity  of  any  density
  $\rho$, i.e.,
  \begin{equation}
  \forall \, \rho, \quad C_{p,\beta,\lambda}[\rho] \ge K_{p,\beta,\lambda}.
  \label{eq:StamDp}
  \end{equation}
  The optimal bound is achieved when, up to a shift and a scaling factor,
  \begin{equation}
    \rho_{p,\beta,\lambda} = u^\vartheta \qquad
    \mbox{with}  \qquad \vartheta = \frac{p^*}{\beta p^* - 1},
  \label{eq:vartheta}
  \end{equation}
  and where $u$ is a solution of the differential equation
  \begin{equation}
  - \frac{d}{dx} \left( \left| \frac{d}{dx} u \right|^{p-2} \frac{d}{dx} u
  \right) + \frac{\gamma}{\vartheta} \: \frac{u^{\lambda \vartheta - 1} -
  u^{\vartheta-1}}{1-\lambda} = 0,
  \label{eq:EDOMin}
  \end{equation}
  with $\gamma$  determined a  posteriori to impose  that $u^\vartheta$  sums to
  unity.   When  $\lambda  \to  1$,  the  limit has  to  be  taken,  leading  to
  $\frac{\gamma}{\vartheta}    \:   \frac{u^{\lambda    \vartheta    -   1}    -
    u^{\vartheta-1}}{1-\lambda} \to \gamma u^{\vartheta-1} \log u$.
\end{Proposition}
\begin{proof}
  The   proof    is   mainly    based   on   the    sharp   Gagliardo--Nirenberg
  inequality~\cite{Agu08},      as      explained      with      details      in
  Appendix~\ref{app:StamDp}.
\end{proof}

Finally,  the minimizers  of the  $(p,\beta,\lambda)$-Fisher--R\'enyi complexity
and  the  tight bound  satisfy  a remarkable  property  of  symmetry, as  stated
hereafter.

\begin{Proposition}\label{prop:Symmetries}
  Let us consider the involutary transform
  \begin{equation}
  \opT_p: \left( \beta , \lambda \right) \mapsto \left( \frac{\beta p^* +
  \lambda - 1}{\lambda p^*} , \frac{1}{\lambda}\right).
  \label{eq:Involution}
  \end{equation}
  The minimizers of the complexity satisfy the relation
  \begin{equation}
  \rho_{p,\opT_p(\beta,\lambda)} \propto \Big[ \rho_{p,\beta,\lambda} \Big]^\lambda,
  \label{eq:RhoInvolution}
  \end{equation}
  and the optimal bounds satisfy the relation
  \begin{equation}
  K_{p,\opT_p(\beta,\lambda)} = \lambda^2 \, K_{p,\beta,\lambda}.
  \label{eq:KInvolution}
  \end{equation}
\end{Proposition}
\begin{proof}
  See Appendix~\ref{app:Symmetries}.
\end{proof}

A difficulty  to determine  the sharp bound  and the  minimizer is to  solve the
nonlinear   differential   equation~\ref{eq:EDOMin}.     One   can   find   in
Corollary~3.2  in~\cite{Agu08}  a  series  of  explicit  equations  allowing  to
determine    the    solution   and    thus    the    optimal    bound   of    in
Equation~\ref{eq:GagliardoNirenberg},  but in  general the  expression  of $u$
remains on an integral form.   Agueh, however, exhibits several situations where
the  solution is  known explicitly  (and  thus the  optimal bound  as well),  as
summarized in the next subsection.


\subsection{Some Explicitly Known Minimizing Distributions}
\label{sec:StamExplicit}

The  particular  cases  are  issued  of  special  cases  of  saturation  of  the
Gagliardo--Nirenberg,  some of them  being studied  by Bercher~\cite{Ber12:06_1,
  Ber12:06_2, Ber13:08} or Lutwak~\cite{LutYan05}.  All these cases are restated
hereafter, with the  notations of the paper. Let us  first recall the definition
of  the  stretched  deformed  Gaussian,  studied  by  Lutwak~\cite{LutYan05}  or
Bercher~\cite{Ber12:06_1,  Ber12:06_2, Ber13:08},  for instance,  also  known as
stretched $q$-Gaussian or stretched Tsallis distributions~\cite{GelTsa04, Tsa09}
and intensively studied in non-extensive physics.

\begin{Definition}[Streched                   deformed                  Gaussian
  distribution]\label{def:StretchedDeformedGaussian}
  Let  $p >  1$ and  $\lambda  > 1-p^*$.   The $(p,\lambda)$-stretched  deformed
  Gaussian distribution is defined by
  \begin{equation}
  \label{eq:StretchedDeformedGaussian}
  g_{p,\lambda}(x) \propto \left\{\begin{array}{lll} \displaystyle \left( 1 +
  (1-\lambda) |x|^{p^*} \right)_+^{\frac{1}{\lambda-1}}, & \mbox{for} & \lambda \ne
  1,\\[4mm]
  \displaystyle \exp\left( - |x|^{p^*} \right), & \mbox{for} &
  \lambda = 1,
  \end{array}\right.
  \end{equation}
  where $(\cdot)_+ =  \max(\cdot,0)$ (the case $\lambda =  1$ is indeed obtained
  taking the limit).
\end{Definition}

This distribution plays  a fundamental role in the  extended Stam inequality, as
we will see in the next subsections and in the next section.


\subsubsection{The Case $\beta = \lambda$}
\label{sec:Bercher}

For any $p > 1$, and for
\begin{equation}
(\beta , \lambda) \in \Bset_p = \left\{
(\beta,\lambda) \in \Dset_p: \quad \beta = \lambda \right\},
\label{eq:Ap}
\end{equation}
one     obtains     that      the     minimizing     distribution     of     the
$(p,\beta,\lambda)$-Fisher--R\'enyi  complexity  is the  $(p,\lambda)$-stretched
deformed Gaussian distribution,
\begin{equation}
\rho_{p,\lambda,\lambda}  \, = \, g_{p,\lambda}
\label{eq:RhoLp}
\end{equation}
(see Corollary  3.4 in~\cite{Agu08}, (i) where  $\lambda = q/s$;  and (ii) where
$\lambda =  s/q$, respectively; the  case $\lambda =  1$ is obtained  taking the
limit  $\lambda  \to  1$  (resp.   lower   and  upper  limit)  or  by  a  direct
computation).   This situation  is nothing  more  than that  studied by  Bercher
in~\cite{Ber12:06_1}  or  Lutwak  in~\cite{LutYan05}.   Remarkably,  by  a  mass
transport approach, Lutwak proved in~\cite{LutYan05} that this relation is valid
for $\lambda > \frac{1}{1+p^*}$, i.e., for
%
\begin{equation}
(\beta , \lambda) \in \Lset_p = \left\{ (\beta,\lambda) \in \Rset_+^{* \, 2}:
\quad \beta = \lambda > \frac{1}{1+p^*} \right\}.
\label{eq:Lp}
\end{equation}
Note that  the exponent of the  Lutwak expression is  not the same as  ours, but
$\beta > 0$ allowing to take the  Lutwak relation to the adequate exponent so as
to obtain our formulation.


\subsubsection{Stretched Deformed Gaussian: The Symmetric Case}
\label{sec:BercherSymmetric}

Immediately, from the  relation Equation~\ref{eq:RhoInvolution} induced by the
involution $\opT_p$,  one obtains,  after a re-parametrization  $\lambda \mapsto
\frac{1}{\lambda}$ and an adequate scaling, for any $p > 1$ and
\begin{equation}
(\beta , \lambda) \in \overline{\Bset}_p =
\left\{ (\beta,\lambda) \in \Dset_p: \quad \beta = \frac{p^*+1-\lambda}{p^*}
\right\}
\label{eq:BpBar}
\end{equation}
that the minimizing distribution is again  a stretched deformed Gaussian,
\begin{equation}
 \rho_{p,\frac{p^*+1-\lambda}{p^*},\lambda} \, = \,  g_{p,2-\lambda}.
\label{eq:RhoLpBar}
\end{equation}
Again, starting from the Lutwak result, the validity of this result extends to
\begin{equation}
(\beta , \lambda) \in \overline{\Lset}_p = \left\{ (\beta,\lambda) \in
\Rset_+^{* \, 2}: \quad 0 < \beta = \frac{p^* + 1 - \lambda}{p^*} <
1+\frac{1}{p^*} \right\},
\label{eq:LpBar}
\end{equation}
and the symmetry of the bound given by Proposition~\ref{prop:Symmetries} remains
valid. Indeed, the minimizers  in $\Lset_p$ satisfying the differential equation
of the Gagliardo--Nirenberg as given in Appendix~\ref{app:StamDp}, the reasoning
of this appendix and of the Appendix~\ref{app:Symmetries} holds.


\subsubsection{Dealing with the Usual Fisher Information}
\label{sec:Agueh}

This situation corresponds to $p = 2$ and $\beta = 1$. Then, for
\begin{equation}
(\beta,\lambda) \in \Aset_2 = \left\{ (\beta,\lambda) \in \Dset_2: \quad \beta =
1 \right\},
\label{eq:A2}
\end{equation}
one obtains the minimizing distribution for $\lambda \ne 1$,
\begin{equation}
\rho_{2,1,\lambda}(x) \propto \left[ \cos\left( \sqrt{1-\lambda} \, |x| \right)
\right]^{\frac{2}{1-\lambda}} \, \un_{\left[ 0 \; \frac{\pi}{2 \,
\Real\{\sqrt{1-\lambda}\}} \right)}(|x|),
\label{eq:RhoA2}
\end{equation}
where $\un_A$  denotes the indicator function  of set $A$,  $\sqrt{-1} = \imath$
(remember  that $\cos(\imath  x)  = \cosh(x)$),  $\Real$  is the  real part  and
$\frac10$ is  to be understood  as $+\infty$ (see Corollary~3.3  in~\cite{Agu06}
with $\lambda  = s/q$  and Corollary~3.4 in~\cite{Agu06}  with $\lambda  = q/s,$
respectively). The  case $\lambda =  1$ is again  obtained by taking  the limit,
leading to  the Gaussian distribution $\rho_{2,1,1}$. (See  previous cases, with
$p = 2$, that corresponds also to the usual Stam inequality.)


\subsubsection{The Symmetrical of  the Usual Fisher Information}
\label{sec:AguehSymmetric}

From  the relation Equation~\ref{eq:RhoInvolution}  induced by  the involution
$\opT_p$, after a re-parametrization  $\lambda \mapsto \frac{1}{\lambda}$ and an
adequate scaling, for $p = 2$ and
\begin{equation}
(\beta,\lambda) \in \overline{\Aset}_2 = \left\{ (\beta,\lambda) \in \Dset_2:
\quad \beta = \frac{\lambda+1}{2} \right\},
\label{eq:A2Bar}
\end{equation}
the minimizing distribution for $\lambda \ne 1$ takes the form
\begin{equation}
\rho_{2,\frac{\lambda+1}{2},\lambda}(x) \propto \left[ \cos\left(
\sqrt{\lambda-1} \, |x| \right) \right]^{\frac{2}{\lambda-1}} \, \un_{\left[ 0
\; \frac{\pi}{2 \, \Real\{\sqrt{\lambda-1}\}} \right)}(|x|)
\label{eq:RhoA2Bar}
\end{equation}
(with, again, the Gaussian as the limit when $\lambda \to 1$).

The graphs  in Figure~\ref{fig:Dp}  describe the domain  $\Dset_p$ (for  a given
$p$).    Therein,   we  also   represent   the   particular  domains   $\Lset_p$
(Bercher--Lutwak situation), $\overline{\Lset}_p$ (transformation of $\Lset_p$),
$\Aset_2$  and  $\overline{\Aset}_2$,  where  the explicit  expressions  of  the
minimizing  distributions  are  known  from  the  works  of~\cite{Agu06,  Agu08,
  LutYan05, Ber12:06_1}.
\begin{figure}[ht!]
\begin{center}
\begin{tikzpicture}
\pgfmathsetmacro{\lM}{2.25}
%
\begin{scope}[scale=1.75]
\pgfmathsetmacro{\p}{5}
\pgfmathsetmacro{\ps}{\p/(\p-1)}
\pgfmathsetmacro{\bM}{1+1/\ps+.5}
\draw[>=stealth,->]  (-.25,0)--({\bM+.2},0)  node[right]{\small  $\beta$};
\draw[>=stealth,->] (0,-.25)--(0,{\lM+.2}) node[above]{\small $\lambda$};
%
%
%
%
\fill[color=black!60!white,opacity=.2]   ({1/\ps},\lM)    --   ({1/\ps},0)   --
({1+1/\ps},1) -- ({1+1/\ps},\lM); 
\draw ({1/\ps},\lM)    --   ({1/\ps},0); 
\draw[thick]  ({1/\ps},0)   -- ({1+1/\ps},1) -- ({1+1/\ps},\lM); 
\draw ({1/\ps+.15},{\lM-.15}) node{\footnotesize $\Dset_p$};
%
%
\draw[dashed,thick] ({1/(1+\ps)},{1/(1+\ps)}) -- ({min(\lM,\bM)},{min(\lM,\bM)});
%
\draw ({min(\lM,\bM)-.125},{min(\lM,\bM)-.3}) node{\footnotesize $\Lset_p$};
%
%
\draw[dashed,thick]  (0,{1+\ps}) -- ({(\ps+1)/\ps},0);
%
\draw ({1+1/\ps},.2) node{\footnotesize $\overline{\Lset}_p$};
%
\draw ({1/(1+\ps)},0) -- ({1/(1+\ps)},-.1) node[below]{\footnotesize $\frac{1}{1+p^*}$};
\draw ({1/\ps},0) -- ({1/\ps},-.1) node[below]{\footnotesize $\frac{1}{p^*}$};
\draw (1,0) -- (1,-.1) node[below]{\footnotesize $1$};
\draw ({1+1/\ps},0) -- ({1+1/\ps},-.1) node[below]{\footnotesize $1+\frac{1}{p^*}$};
\draw (0,{1/(1+\ps)}) -- (-.1,{1/(1+\ps)}) node[left]{\footnotesize $\frac{1}{1+p^*}$};
\draw (0,1) -- (-.1,1) node[left]{\footnotesize $1$};
\draw (0,{1+\ps}) -- (-.1,{1+\ps}) node[left]{\footnotesize $1+p^*$};
\draw ({\bM/2+.2},-.6) node{\small (a)};
\end{scope}
%
%
\begin{scope}[xshift=7.5cm,scale=1.75]
\pgfmathsetmacro{\p}{2}
\pgfmathsetmacro{\ps}{\p/(\p-1)}
\pgfmathsetmacro{\bM}{1+1/\ps+.5}
\draw[>=stealth,->]  (-.25,0)--({\bM+.2},0)  node[right]{\small  $\beta$};
\draw[>=stealth,->] (0,-.25)--(0,{\lM+.2}) node[above]{\small $\lambda$};
%
%
%
%
\fill[color=black!60!white,opacity=.2]   ({1/\ps},\lM)    --   ({1/\ps},0)   --
({1+1/\ps},1) -- ({1+1/\ps},\lM); 
\draw ({1/\ps},\lM)    --   ({1/\ps},0); 
\draw[thick]  ({1/\ps},0)   -- ({1+1/\ps},1) -- ({1+1/\ps},\lM); 
\draw ({1/\ps+.15},{\lM-.15}) node{\footnotesize $\Dset_2$};
%
%
\draw[dashed,thick] (1,.5) -- (1,{\lM});
\draw (1.14,{\lM-.15}) node{\footnotesize $\Aset_2$};
%
\draw[dashed,thick] (.5,0) -- (1.5,2);
\draw (1.37,1.53) node{\footnotesize $\overline{\Aset}_2$};
%
%
%
\draw ({1/\ps},0) -- ({1/\ps},-.1) node[below]{\footnotesize $\frac12$};
\draw (1,0) -- (1,-.1) node[below]{\footnotesize $1$};
\draw ({1+1/\ps},0) -- ({1+1/\ps},-.1) node[below]{\footnotesize $\frac32$};
\draw (0,1) -- (-.1,1) node[left]{\footnotesize $1$};
\draw ({\bM/2+.2},-.6) node{\small (b)};
\end{scope}
\end{tikzpicture}
\end{center}
\caption{(\textbf{a}) the domain $\Dset_p$ for a given $p$ is represented by the
  gray area  (here $p > 2$).  The  thick line belongs to  $\Dset_p$.  The dashed
  line  represents   $\Lset_p$,  corresponding   to  the  Lutwak   situation  of
  Section~\ref{sec:Bercher},  where the  relation holds  and the  minimizers are
  explicitly   known  (stretched   deformed  Gaussian   distributions),  whereas
  $\overline{\Lset}_p$    corresponds    to   Section~\ref{sec:BercherSymmetric}
  ($\Bset_p$  and  $\overline{\Bset}_p$  obtained  by  the  Gagliardo--Nirenberg
  inequality are  their restrictions to $\Dset_p$);  (\textbf{b}) same situation
  for  $p =  2$, with  the  domains $\Aset_2$  and $\overline{\Aset}_2$  (dashed
  lines)  that  correspond to  the  situations  of Sections~\ref{sec:Agueh}  and
  ~\ref{sec:AguehSymmetric},  respectively, ($\Lset_2$  and $\overline{\Lset}_2$
  are not represented for the clarity of the figure).}
\label{fig:Dp}
\end{figure}
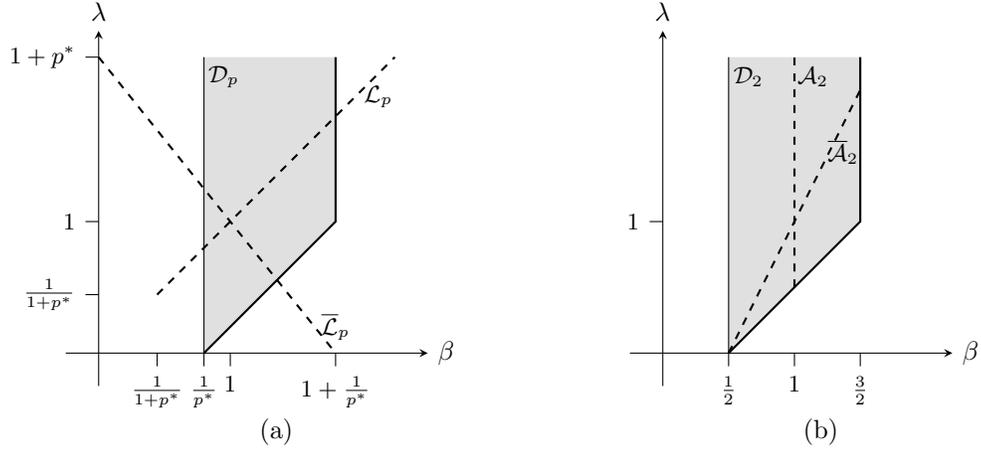


\section{Extended  Optimal Stam  Inequality: A Step Further}
\label{sec:ExtendedStam}

In  this section,  we further  extend the  previous Stam  inequality,  namely by
largely  widening  the domain  for  the  parameters  and disentangling  the  two
connected  parameters.    For  this,  we   use  the  \textit{differential-escort
  deformation} introduced in~\cite{PueRud17}, which is the key tool allowing for
introducing a  new degree of freedom.   Afterwards, we will  give the minimizing
distribution that results in a new deformation of the Gaussian family intimately
linked with the inverse incomplete beta functions.


\subsection{Differential-Escort Distribution: A Brief Overview}

We  have  already realized  the  crucial  role that  the  power  operation of  a
probability density  function $\rho$ plays.  The  subsequent escort distribution
duly normalized,  $\displaystyle \frac{\rho(x)^\alpha}{\int_\Rset \rho(x)^\alpha
  dx}$,   is    a   simple    monoparametric   deformation   of    $\rho$   (see
e.g.,~\cite{Ber10}).  Notice  that the parameter  $\alpha$ allows us  to explore
different  regions of  $\rho$, so  that,  for $\alpha  > 1$,  the more  singular
regions  are amplified  and, for  $\alpha <  1$, the  less singular  regions are
magnified.  A  careful look  at the minimizing  distributions of the  usual Stam
inequality  shows that the  $x$-axis is  stretched via  a power  operation. This
makes us guess that  a certain nonlinear stretching may also play  a key role in
the saturation (i.e., equality) of the extended Stam inequality.

These ideas led us to  the definition of the differential-escort distribution of
a probability  distribution $\rho$ (see also~\cite{PueRud17}),  motivated by the
following principle.  The power operation  provokes a two-fold stretching in the
density  itself  and  in  the  differential  interval  so  as  to  conserve  the
probability in the differential intervals: $\rho_\alpha(y) dy = \rho(x) dx$ with
$\rho_\alpha(y) = \rho(x(y))^\alpha$.

\begin{Definition}[Differential-escort
  distributions]\label{def:differential-escort_def}
  Given  a probability  distribution $\rho(x)$  and given  an index  $\alpha \in
  \Rset$, the  differential-escort distribution of  $\rho$ of order  $\alpha$ is
  defined as
  \begin{equation}
  \opE_\alpha[\rho](y) = \Big[ \rho\big( x(y) \big) \Big]^\alpha,
  \end{equation}
  where  $y(x)$ is  a bijection  satisfying  \ $\frac{dy}{dx}  = \left[  \rho(x)
  \right]^{1-\alpha}$ \ and $y(0) = 0$.
\end{Definition}

The differential-escort transformation $\opE_\alpha$ exhibits various properties
studied in  detail in~\cite{PueRud17}.  We  present here the key  ones, allowing
the extension of the Stam inequality in a wider domain than that of the previous
section.

\begin{Property}\label{prop:PropertiesStretchedEscort}
  The differential-escort transformation satisfies the composition relation
  \begin{equation}
  \opE_\alpha \circ \opE_{\alpha'} = \opE_{\alpha'} \circ \opE_\alpha =
  \opE_{\alpha \alpha'}
  \end{equation}
  where $\circ$  is the composition  operator.  Moreover, since $\opE_1$  is the
  identity, for any $\alpha \ne 0$, $\opE_\alpha$ is invertible and,
  \begin{equation}
  \opE_\alpha^{-1} = \opE_{\alpha^{-1}}.
  \end{equation}
\end{Property}

In  addition  to   the  trivial  case  $\alpha  =   1$,  keeping  invariant  the
distribution, a remarkable case is given by $\alpha = 0$, leading to the uniform
distribution. This case is non surprising since then $x(y)$ is nothing more than
the  inverse of  the cumulative  density function,  well known  to  uniformize a
random~vector~\cite{Dev86}.

In the  sequel, we focus  on the differential-escort distributions  obtained for
$\alpha > 0$.  Under this condition, when $\rho$ is continuously differentiable,
its differential-escort is also  continuously differentiable.  This is important
to   be    able   to   define   its    $(p,\lambda)$-Fisher   information   (see
Definition~\ref{def:biparametricFisher}).     Under    this    condition,    the
differential-escort transformation  induces a scaling  property on the  index of
the R\'enyi entropy power (for this quantity it remains true for any $\alpha \in
\Rset$),  the  $(p,\beta)$-Fisher  information,   and  thus  on  the  subsequent
complexity as stated in the following proposition.
\begin{Proposition}\label{prop:ComplexityStretchedEscort}
  Let a  probability distribution $\rho$  and an index  $\alpha > 0$.  Then, the
  R\'enyi  entropy powers  of  $\rho$ and  its differential-escort  distribution
  $\opE_\alpha[\rho]$ satisfy that
  \begin{equation}
  N_\lambda\Big[\opE_\alpha[\rho]\Big] = \Big( N_{1+\alpha(\lambda-1)}[\rho]
  \Big)^\alpha
  \label{eq:AffineRenyi}
  \end{equation}
  for  any  $\lambda  \in  \Rset_+^*$.   Moreover,  if  the  density  $\rho$  is
  continuously differentiable,  then the extended Fisher  information of $\rho$
  and its differential-escort distribution $\opE_\alpha[\rho]$ satisfy that
  \begin{equation}
  F_{p,\beta}\Big[\opE_\alpha[\rho]\Big] = \alpha^{\frac{2}{\beta}} \Big(
  F_{p,\alpha\beta}[\rho] \Big)^\alpha
  \label{eq:AffineFisher}
  \end{equation}
  for any $p > 1, \beta \in \Rset_+^*$.

  Consequently, the $(p,\beta,\lambda)$-Fisher--R\'enyi complexity of $\rho$ and
  of $\opE_\alpha[\rho]$ satisfy the relation
  \begin{equation}
  C_{p,\beta,\lambda}\Big[\opE_\alpha[\rho] \Big] = \alpha^2 \, C_{p ,
  \opA_\alpha(\beta , \lambda)}[\rho].
  \label{eq:AffineComplexity}
  \end{equation}
\end{Proposition}
\begin{proof}
  It is straightforward to note that
  \begin{eqnarray*}
  \Big( N_\lambda\Big[\opE_\alpha[\rho]\Big]
  \Big)^{\frac{1-\lambda}{2}}
  & = & \int_\Rset  \left[  \opE_\alpha[\rho](y) \right]^\lambda  dy\\[1mm]
  & = &\int_\Rset \left[ \opE_\alpha[\rho]\big( y(x) \big) \right]^\lambda
  \frac{dy}{dx} \, dx\\[1mm]
  & = &\int_\Rset \left[ \rho(x) \right]^{\alpha \lambda + 1 - \alpha} dx\\[1mm]
  & = &\left(  N_{1+\alpha(\lambda-1)}[\rho]
  \right)^{\frac{\alpha(1-\lambda)}{2}},
  \end{eqnarray*}
  leading to Equation~\ref{eq:AffineRenyi}.

  Similarly,
  \begin{eqnarray*}
  \Big(  F_{p,\beta}\Big[\opE_\alpha[\rho]\Big]  \Big)^{\frac{p \beta}{2}} & = &
  \int_\Rset \Big| \left[ \opE_\alpha[\rho](y) \right]^{\beta-2} \frac{d}{dy}
  \left[ \opE_\alpha[\rho](y) \right] \Big|^p \opE_\alpha[\rho](y) \, dy\\[1mm]
  & = & \int_\Rset \Big| \left[ \opE_\alpha[\rho](y(x)) \right]^{\beta-2}
  \frac{d}{dx} \left[ \opE_\alpha[\rho](y(x)) \right] \frac{dx}{dy} \Big|^p
  \opE_\alpha[\rho](y(x)) \, \frac{dy}{dx} \, dx\\[1mm]
  & = & \int_\Rset \Big| \left[ \rho(x) \right]^{\alpha (\beta-2)} \frac{d}{dx}
  \left[ \left( \rho(x) \right)^\alpha \right] \left[ \rho(x) \right]^{\alpha-1}
  \Big|^p \rho(x) \, dx\\[1mm]
  & = & \int_\Rset \Big| \alpha \left[ \rho(x) \right]^{\alpha \beta-2}
  \frac{d}{dx} \left[ \rho(x)\right] \Big|^p \rho(x) \, dx,
  \end{eqnarray*}
  leading to Equation~\ref{eq:AffineFisher}.

  Relation~\ref{eq:AffineComplexity}       is      a       consequence      of
  Equations~\ref{eq:AffineRenyi}   and  ~\ref{eq:AffineFisher}   together  with
  Definition~\ref{def:triparametricComplexity} of  the complexity.
\end{proof}
One may mention~\cite{Kor17} where the  author studies the effect of a rescaling
of the Tsallis  non-additive parameter, equivalent to the  entropic parameter of
the     R\'enyi     entropy,     and     that     is     exactly     that     of
Equation~\ref{eq:AffineRenyi}. In particular, this  rescaling has an effect on
the maximum entropy distribution in such  a way that it is equivalent to elevate
this particular distribution to a power.  Here, the spirit is slightly different
since we start from a given distribution and the nonlinear stretching is made on
the  state ($x$-axis)  of  any probability  density in  such  a way  that it  is
elevated  to  an   exponent.   The  stretching  is  intimately   linked  to  the
distribution, being  of maximum entropy  or not, and  the scaling effect  on the
R\'enyi is a  consequence of this nonlinear stretching.  The  study of the links
between the present result and that of~\cite{Kor17} goes beyond the scope of our
work and remains as a perspective.


\subsection{Enlarging the Validity Domain of the Extended Stam Inequality}
\label{sec:EnlargedStam}

We have now  all the ingredients to  enlarge the domain of validity  of the Stam
inequality.  Moreover,  we are able to  determine an explicit  expression of the
minimizer by the mean of a special function, i.e., more simple to determine than
as in Proposition~\ref{prop:StamDp}, and of the tight bound as well.

To this aim, let us consider the following affine transform $\opA_a$ and the set
of transformation for $a \in \Rset_+^*$,
\begin{equation}
\opA_a : (\beta,\lambda) \mapsto (a \beta , 1 + a (\lambda-1)) \quad \mbox{and}
\quad \opA(\beta,\lambda) = \left\{ \opA_a(\beta,\lambda): \: a \in \Rset_+^*
\right\} \, \cap \, \Rset_+^{* \, 2}.
\end{equation}
Then,    for    any   strictly    positive    real    $a$,    one   can    apply
Proposition~\ref{prop:StamDp}  to   $\opE_a[\rho]$,  that  is,  for   $p  >  1$,
$(\beta,\lambda) \in \Dset_p$,  $C_{p,\beta,\lambda}\big[ \opE_a[\rho] \big] \ge
K_{p,\beta,\lambda}$.                          Thus,                        from
Proposition~\ref{prop:ComplexityStretchedEscort},   one  immediately   has  that
$C_{p,\opA_a(\beta,\lambda)}[\rho]   \ge   a^{-2}   K_{p,\beta,\lambda}   \equiv
K_{p,\opA_a(\beta,\lambda)}$.  Moreover,  this inequality  is sharp since  it is
achieved   for  $\opE_a[\rho]   =   \rho_{p,\beta,\lambda}$,  \   i.e.,  \   for
$\rho_{p,\opA_a(\beta,\lambda)} = \opE_{a^{-1}}[\rho_{p,\beta,\lambda}]$

As a conclusion, the existence of a universal optimal positive constant bounding
the  complexity (see  Proposition~\ref{prop:StamDp}) extends  from  $\Dset_p$ to
$\opA(\Dset_p)$.   Note that $\opA(\beta,\lambda)$  is the  overlap of  the line
defined by  the point  $(0,1)$ and $(\beta,\lambda)$  itself (achieved for  $a =
1$), and $\Rset_+^{* \, 2}$,  as depicted Figure~\ref{fig:DpTilda}.  Then, it is
straightforward to see that $ \widetilde{\Dset}_p \equiv \opA(\Dset_p) = \left\{
  (\beta,\lambda) \in  \Rset_+^{* \, 2}: \: \lambda  > 1 - \beta  p^* \right\} $
(see Figure~\ref{fig:DpTilda}a).  The approach is thus the following:
\begin{itemize}
\item  Consider a point  $(\beta,\lambda) \in  \widetilde{\Dset}_p$ and  find an
  index  $\alpha  \in   \Rset_+^*$  such  that  $\opA_\alpha(\beta,\lambda)  \in
  \Dset_p$, which is  a point of the intersection between  $\Dset_p$ and the line
  joining $(0,1)$ and $(\beta,\lambda)$.
\item      Apply      Proposition~\ref{prop:StamDp}      for      the      point
  $(p,\opA_\alpha(\beta,\lambda))$,   leading  to  the   minimizing  distribution
  $\rho_{p,\opA_\alpha(\beta,\lambda)}$ and its corresponding bound.
\item       Then,      remarking       that       $\opA_{\alpha^{-1}}      \circ
  \opA_{\alpha}(\beta,\lambda) = (\beta,\lambda)$, the minimizer of the extended
  complexity     writes    $\rho_{p,\beta,\lambda}    =     \opE_\alpha    \Big[
  \rho_{p,\opA_\alpha(\beta,\lambda)} \Big]$ and  the corresponding bound can be
  computed from  this minimizer or  noting that $K_{p,\beta,\lambda}  = \alpha^2
  K_{p,\opA_\alpha(\beta,\lambda)}$.
\end{itemize}

The same  procedure obviously applies  dealing with $\Lset_p$:  $\opA(\Lset_p) =
\left\{ (\beta,\lambda) \in \Rset_+^{* \, 2}: \: 1-\beta p^* < \lambda < \beta+1
\right\}$   appears    to   be   a   subset    of   $\widetilde{\Dset}_p$   (see
Figure~\ref{fig:DpTilda}b).     Similarly,   one    can   also    deal   with
$\overline{\Lset}_p$:  $\opA(\overline{\Lset}_p) =  \left\{  (\beta,\lambda) \in
  \Rset_+^{*  \, 2}: \:  \lambda >  1 -  \frac{p^* \beta}{p^*+1}  \right\}$ also
appears     to     be     a     subset     of     $\widetilde{\Dset}_p$     (see
Figure~\ref{fig:DpTilda}c).  Remarkably,  $\opA(\Dset_p) = \opA(\Lset_p) \cup
\opA(\overline{\Lset}_p)$.   Moreover,  we  have  explicit expressions  for  the
minimizers  in both $\Lset_p$  and $\overline{\Lset}_p$,  which greatly  eases
determining the minimizers in $\widetilde{\Dset}_p$ (including $\Dset_p$ itself).

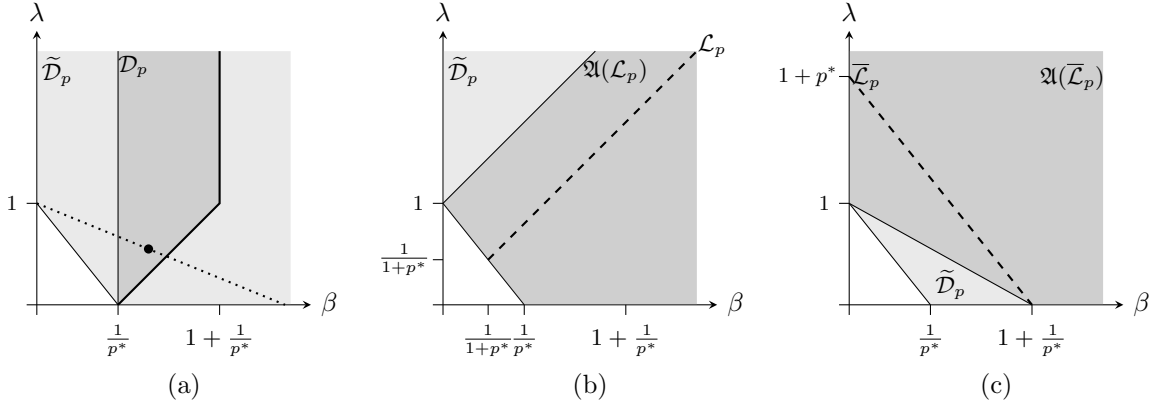
\begin{figure}[ht!]
\begin{center}
\begin{tikzpicture}
%
%
\pgfmathsetmacro{\p}{5}
\pgfmathsetmacro{\ps}{\p/(\p-1)}
\pgfmathsetmacro{\lM}{2.5}
\pgfmathsetmacro{\bM}{2.5}
\pgfmathsetmacro{\bz}{1.1} 
\pgfmathsetmacro{\lz}{.55} 
%
%
\begin{scope}[scale=1.35]
%
%
\draw[>=stealth,->] (-.1,0) -- ({\bM+.2},0) node[right]{\small $\beta$};
\draw[>=stealth,->] (0,-.1) -- (0,{\lM+.2}) node[above]{\small $\lambda$};
%
\fill[color=black!40!white,opacity=.2]
(0,\lM) -- (0,1) -- ({1/\ps},0) -- (\bM,0) -- (\bM,\lM); 
\draw (0,1) -- ({1/\ps},0); 
\draw (+.2,{\lM-.2}) node {\footnotesize $\widetilde{\Dset}_p$};
%
%
\fill[color=black!60!white,opacity=.2]
( {1/\ps} , \lM ) -- ( {1/\ps} , 0 ) -- ( {1+1/\ps} , 1 ) -- ( {1+1/\ps} , \lM );
\draw ({1/\ps},\lM)    --   ({1/\ps},0); 
\draw[thick]  ({1/\ps},0)   -- ({1+1/\ps},1) -- ({1+1/\ps},\lM); 
\draw ({1/\ps+.15},{\lM-.15}) node {\footnotesize $\Dset_p$};
%
\draw (\bz,\lz) node[scale=3]{.};
\draw[dotted,thick] (0,1) -- ({min(\bz/(1-\lz),\bM)},{max(0,1+(\lz-1)*\bM/\bz)});
%
\draw ({1/\ps},0) -- ({1/\ps},-.1) node[below]{\footnotesize $\frac{1}{p^*}$};
\draw ({1+1/\ps},0) -- ({1+1/\ps},-.1) node[below]{\footnotesize $1+\frac{1}{p^*}$};
\draw (0,1) -- (-.1,1) node[left,scale=.9]{\footnotesize $1$};
\draw ({\bM/2+.2},-.8) node{\small (a)};
\end{scope}
%
%
%
\begin{scope}[scale=1.35,xshift=4cm]
\draw[>=stealth,->] (-.1,0) -- ({\bM+.2},0) node[right]{\small $\beta$};
\draw[>=stealth,->] (0,-.1) -- (0,{\lM+.2}) node[above]{\small $\lambda$};
%
\fill[color=black!40!white,opacity=.2]
(0,\lM) -- (0,1) -- ({1/\ps},0) -- (\bM,0) -- (\bM,\lM); 
\draw (0,1) -- ({1/\ps},0); 
\draw (.2,{\lM-.2}) node {\footnotesize $\widetilde{\Dset}_p$};
%
%
\fill[color=black!60!white,opacity=.2]
(0,1) -- ({1/\ps},0) -- (\bM,0) --
(\bM,{min(1+\bM,\lM)}) --
({min(\lM-1,\bM)},{min(\lM,\bM+1)}); 
\draw (0,1)--({min(\lM-1,\bM)},{min(\lM,\bM+1)}); 
\draw ({min(\lM-1,\bM)+.2},{min(\lM,\bM+1)-.2}) node {\footnotesize $\opA(\Lset_p)$};
%
%
\draw[dashed,thick] ({1/(\ps+1)},{1/(\ps+1)})--({min(\lM,\bM)},{min(\lM,\bM)});
\draw ({min(\lM,\bM)+.15},{min(\lM,\bM)+.05}) node {\footnotesize $\Lset_p$};
\draw ({1/(1+\ps)},0) -- ({1/(1+\ps)},-.1) node[below]{\footnotesize $\frac{1}{1+p^*}$};
\draw ({1/\ps},0) -- ({1/\ps},-.1) node[below]{\footnotesize $\frac{1}{p^*}$};
\draw ({1+1/\ps},0) -- ({1+1/\ps},-.1) node[below]{\footnotesize $1+\frac{1}{p^*}$};
\draw (0,{1/(1+\ps)}) -- (-.1,{1/(1+\ps)});
\draw (-.025,{1/(1+\ps)}) node[left]{\footnotesize $\frac{1}{1+p^*}$};
\draw (0,1) -- (-.1,1) node[left,scale=.9]{\footnotesize $1$};
\draw ({\bM/2+.2},-.8) node{\small (b)};
\end{scope}
%
%
%
\begin{scope}[scale=1.35,xshift=8cm]
\pgfmathsetmacro{\lm}{-1}
\pgfmathsetmacro{\bm}{-1.5}
\draw[>=stealth,->] (-.1,0) -- ({\bM+.2},0) node[right]{\small $\beta$};
\draw[>=stealth,->] (0,-.1) -- (0,{\lM+.2}) node[above]{\small $\lambda$};
%
\fill[color=black!40!white,opacity=.2]
(0,\lM) -- (0,1) -- ({1/\ps},0) -- (\bM,0) -- (\bM,\lM); 
\draw (0,1) -- ({1/\ps},0); 
\draw ({1/\ps+.2},.2) node {\footnotesize $\widetilde{\Dset}_p$};
%
%
\fill[color=black!60!white,opacity=.2]
(0,1) -- ({1+1/\ps},0) -- (\bM,0) -- (\bM,\lM) -- (0,\lM); 
\draw (0,1) -- ({1+1/\ps},0);
\draw ({\bM-.3},{\lM-.25}) node {\footnotesize $\opA(\overline{\Lset}_p)$};
%
%
\draw[dashed,thick]  (0,{1+\ps}) -- ({(\ps+1)/\ps},0);
\draw (.17,{1+\ps}) node{\footnotesize $\overline{\Lset}_p$};
\draw ({1/\ps},0) -- ({1/\ps},-.1) node[below]{\footnotesize $\frac{1}{p^*}$};
\draw ({1+1/\ps},0) -- ({1+1/\ps},-.1) node[below]{\footnotesize $1+\frac{1}{p^*}$};
\draw (0,1) -- (-.1,1) node[left,scale=.9]{\footnotesize $1$};
\draw (0,{1+\ps}) -- (-.1,{1+\ps});
\draw (-.025,{1+\ps}) node[left,scale=.9]{\footnotesize $1+p^*$};
\draw ({\bM/2+.2},-.8) node{\small (c)};
\end{scope}
\end{tikzpicture}
\end{center}
\caption{Given a $p$, the domain in gray represents $\widetilde{\Dset}_p$, where
  we know  that the $(p,\beta,\lambda)$-Fisher--R\'enyi  complexity is optimally
  lower   bounded    and   where   the   minimizers   can    be   deduced   from
  proposition~\ref{prop:StamDp}.   (\textbf{a})~the   domain  in  dark   gray  represents
  $\Dset_p$, which is obviously included in $\widetilde{\Dset}_p$; the  dot is a
  particular point $(\beta,\lambda) \in  \Dset_p$ and the dotted line represents
  its  transform   by  $\opA$;   (\textbf{b})  the  domain   in  dark   gray  represents
  $\opA(\Lset_p) \subset \widetilde{\Dset}_p$, which obviously contains $\Lset_p$
  represented by the dashed line;  (\textbf{c}) same as~(\textbf{b}) with $\overline{\Lset}_p$ and
  $\opA(\overline{\Lset}_p) \subset \widetilde{\Dset}_p$.  This illustrates that
  $\widetilde{\Dset}_p = \opA(\Lset_p) \cup \opA(\overline{\Lset}_p)$. }
\label{fig:DpTilda}
\end{figure}

These remarks, together with both  the knowledge of the minimizing distributions
and  the  bound on  $\Lset_p  \cup  \overline{\Lset}_p$,  lead to  the  following
definition and proposition.

\begin{Definition}[$(p,\beta,\lambda)$-Gaussian
  distribution]\label{def:triparametricGaussian}
  For any  $p >  1$ and $(\beta,\lambda)  \in \Rset_+^{*  \, 2}$, we  define the
  $(p,\beta,\lambda)$-Gaussian distribution as
 \begin{equation}
  g_{p,\beta,\lambda}(x) \propto \left\{\begin{array}{lcl} \left[ 1 - \Beta^{-1}
  \! \left( \frac{1}{p^*} , q_{p,\beta,\lambda} ; \frac{p^*
  |x|}{|1-\lambda|^{\frac{1}{p^*}}} \right)\right]^{\frac{1}{|1-\lambda|}}
  \un_{\left[ 0 \; B\left( \frac{1}{p^*} , q_{p,\beta,\lambda} \right)
  \right]} \! \left( \frac{p^* |x|}{|1-\lambda|^{\frac{1}{p^*}}} \right),
  & \mbox{if} & \lambda \ne 1,\\[7.5mm]
  \exp\left( - \frac{\G^{-1} \left( \frac{1}{p^*} \; \left(
 \frac{\beta-1}{\beta} \right)^{\frac{1}{p^*}} p^* |x| \right)}{\beta-1} \right) \,
 \un_{\left[ 0 \; \frac{\Gamma( 1/p^*)}{\un_{(0 \; 1)}(\beta)}
 \right]}\big( p^* |x| \big),
  & \mbox{if} & \begin{array}{l}\lambda = 1,\\ \beta \ne 1, \end{array}\\[7.5mm]
  \exp\left( - |x|^{p^*} \right),
  & \mbox{if} & \beta = \lambda = 1,
  \end{array}\right.
 \end{equation}
 with
  %
  \begin{equation}
  q_{p,\beta,\lambda} = \frac{\beta-1}{|1-\lambda|} +
  \frac{\un_{\Rset_+}(1-\lambda)}{p}.
  \end{equation}
  $\opT_p$ is  the involutary transform  defined Equation~\ref{eq:Involution}.
  $\displaystyle  \Beta(a,b,x)  =  \int_0^x   t^{a-1}  (1-t)^{b-1}  dt$  is  the
  incomplete  beta function,  defined when  $a >  0$ and  for $x  \in [0  \; 1)$
  (see~\cite{OlvLoz10}),   and   $\displaystyle   B(a,b)   =  \lim_{x   \to   1}
  \Beta(a,b,x)$, that  is the  standard beta  function if $b  > 0$  and infinite
  otherwise.     $\Beta^{-1}$   is    thus   the    inverse    incomplete   beta
  function. Finally,  $\displaystyle \G(a,x) = \int_0^x t^{a-1}  \exp(-t) \, dt$
  is  the incomplete  gamma  function,  defined when  $a  > 0$  and  for $x  \in
  \Rset$~\cite{OlvLoz10}, and $\Gamma(a) =  \lim_{x \to +\infty} \G(a,x)$ is the
  gamma  function.   By  definition,  $z^\alpha =  |z|^\alpha  e^{\imath  \alpha
    \textit{Arg}(t)}$ where  $0 \le  \textit{Arg}(t) <  2 \pi$.  Finally,  by convention  $1/0 =
  +\infty$.
\end{Definition}
Note that, when $b > 0$, the  inverse incomplete beta function is well known and
tabulated in the usual mathematical softwares since it is the inverse cumulative
function   of  the  beta   distributions~\cite{AbrSte70}.   Otherwise,   as  the
incomplete     beta     function     writes    through     an     hypergeometric
function~\cite{GraRyz07}  (see also~\cite{OlvLoz10,AbrSte70}),  also  well known
and  tabulated,  $\Beta^{-1}$  can   be  at  least  numerically  computed.   The
incomplete   beta  function   contains   many  special   cases  for   particular
parameters~\cite{GraRyz07,PruBry90:v3}.  For instance,  when $a+b$ is a negative
integer, they express as elementary functions~\cite{GraRyz07}.

Similarly, when its argument is  positive, the incomplete gamma function and its
inverse are well  known and tabulated because they are  linked to the cumulative
distribution of  gamma laws~\cite{AbrSte70}.   Even for negative  arguments, the
incomplete  gamma function  is very  often tabulated  in  mathematical software.
Otherwise,    one   can    write   it    using   a    confluent   hypergeometric
function~\cite{OlvLoz10}    (see    also~\cite{GraRyz07,AbrSte70}),    generally
tabulated. Thus, it can be  inverted at least numerically.  The incomplete gamma
function also  contains special cases for particular  parameters.  For instance,
$\G\left(\frac12,x^2\right) = \textrm{erf}(x),$ where $\textrm{erf}$
is the  error function~\cite{OlvLoz10}.  Hence, for  $p = 2$ and  $\lambda = 1$,
the $(p,\beta,\lambda)$-Gaussian writes in terms of the inverse error function.

Now, from the procedure previously described, we obtain the Stam inequality with
the widest possible  domain, together with the minimizing  distributions and the
explicit tight lower bound.

\begin{Proposition}[Stam inequality in a wider domain]\label{prop:StamDpTilda}
  The  $(p,\beta,\lambda)$-Fisher--R\'enyi  complexity  is non  trivially  lower
  bounded as follows:
  \begin{equation}
  \forall \, p > 1, \quad (\beta,\lambda) \in \widetilde{\Dset}_p = \left\{
  (\beta,\lambda) \in \Rset_+^{* \, 2}: \: \lambda > 1 - \beta p^* \right\},
  \qquad C_{p,\beta,\lambda}[\rho] \ge K_{p,\beta,\lambda}.
  \end{equation}
  The minimizers are explicitly given by
  \begin{equation}
  \textrm{argmin}_\rho  C_{p,\beta,\lambda}[\rho]   =  g_{p,\beta,\lambda},
  \end{equation}
  the                       $(p,\beta,\lambda)$-Gaussian                      of
  Definition~\ref{def:triparametricGaussian}.  Proposition~\ref{prop:Symmetries}
  remains valid in $\widetilde{\Dset}_p$. Moreover, the tight bound is
  \begin{equation}
  K_{p,\beta,\lambda} = \left\{\begin{array}{lll}
  \left( \frac{2}{p^* \zeta_{p,\beta,\lambda}} \left( \frac{p^*
  \zeta_{p,\beta,\lambda}}{|1-\lambda|} \right)^{\frac{1}{p^*}} \left( \frac{p^*
  \zeta_{p,\beta,\lambda}}{p^* \zeta_{p,\beta,\lambda} -|1-\lambda|}
  \right)^{\frac{\zeta_{p,\beta,\lambda}}{|1-\lambda|}+\frac{1}{p}} \, B\left(
  \frac{1}{p^*} , \frac{\zeta_{p,\beta,\lambda}}{|1-\lambda|} + \frac{1}{p}
  \right) \right)^2, & \mbox{if} & \lambda \ne 1,\\[7.5mm]
  \left( \frac{2 \, e^{\frac{1}{p^*}} \, \Gamma\left( \frac{1}{p^*}
  \right)}{\beta p^{* \, \frac{1}{p}}} \right)^2, & \mbox{if} & \lambda = 1,
  \end{array}\right.
  \end{equation}
  with
  \begin{equation}
  \zeta_{p,\beta,\lambda} = \beta + \frac{(\lambda-1)_+}{p^*}.
  \label{eq:zeta}
  \end{equation}
\end{Proposition}
\begin{proof}
  See Appendix~\ref{app:StamDpTilda}.
\end{proof}


\section{Applications to Quantum Physics}
\label{sec:Applications}

Let  us now  apply the  $(p,\beta,\lambda)$-Fisher--R\'enyi complexity  for some
specific values of the parameters to  the analysis of the two main prototypes of
$\d$-dimensional  quantum  systems  subject  to  a  central  (i.e.,  spherically
symmetric)   potential;    namely,   the   hydrogenic    and   harmonic   (i.e.,
oscillator-like) systems.  The wave functions  of the bound stationary states of
these systems  have the same  angular part, so  that we concentrate here  on the
radial distribution in both position and momentum spaces.


\subsection{Brief Review on the Quantum Systems with Radial Potential}

The  time-independent Schr\"odinger equation  of a  single-particle system  in a
central potential $V(r)$ can be written as
\begin{equation}
\left( - \frac12 \vec{\nabla}^2_\d + V(r) \right) \Psi \left( \vr \right) = E_n
\, \Psi \left( \vr \right),
\label{eq:schrodinger_central}
\end{equation}
(atomic units are  used from here onwards), where  $\vec{\nabla}_\d$ denotes the
$\d$-dimensional  gradient operator  and the  position vector  \ $\vr  =  (x_1 ,
\ldots , x_\d)$ \  in hyperspherical units is given by \  $\left( r , \theta_1 ,
  \theta_2  , \ldots  , \theta_{\d-1}  \right) \equiv  \left( r  , \Omega_{\d-1}
\right)$, \ $\Omega_{\d-1}\in \Sset^{\d-1}$  \ the unit $\d$-dimensional sphere,
where \ $r \equiv \left| \vr \right| = \sqrt{\sum_{i=1}^\d x_i^2} \in \Rset_+$ \
and \ $\displaystyle x_i = r \left( \prod_{k=1}^{i-1} \sin \theta_k \right) \cos
\theta_i$ \ for \ $1 \le i \le \d$  \ and with \ $\theta_i \in [0 \; \pi)$ \ for
$i < \d-1$, \ $\theta_{\d-1} \equiv \phi \in  [0 \; 2 \pi)$ \ and \ $\theta_\d =
0$ \  by convention.  The  physical wave functions  are known to  factorize (see
e.g.,~\cite{Nie79, YanVan94, Ave02}) as
\begin{equation}
\Psi_{n,l, \left\lbrace \mu \right\rbrace }(\vr) = \Rad_{n,l}(r) \,
{\cal{Y}}_{l,\{\mu\}}(\Omega_{\d-1}),
\label{eq:Psi_RY}
\end{equation}
where $\Rad_{n,l}(r)$ and  $\Hyp_{l,\{\mu\}}\left( \Omega_{\d-1} \right)$ denote
the radial  and the angular part,  respectively, being $\left(  l , \left\lbrace
    \mu \right\rbrace  \right) \equiv \left( l  \equiv \mu_1 , \mu_2  , \ldots ,
  \mu_{\d-1}  \right)$  the  hyperquantum  numbers  associated  to  the  angular
variables   $\Omega_{\d-1}\equiv  \left(   \theta_1  ,   \theta_2  ,   \ldots  ,
  \theta_{\d-1}  \right)$,  which  may  take  all  values  consistent  with  the
inequalities  $l \equiv  \mu_1 \geq  \mu_2  \geq \ldots  \geq \left|  \mu_{\d-1}
\right| \equiv \left| m \right| \geq 0$.

As already  stated, the  angular part $\Hyp_{l,\{\mu\}}$  is independent  of the
potential  $V$  and  its  expression is  detailed  in~\cite{Ave02,  DehLop10:07,
  YanVan99,  SanZoz11}, for  instance.   Only the  radial  part $\Rad_{n,l}$  is
dependent  on $V$ (and  also on  the energy  level $n$  and the  angular quantum
number $l$), being the solution of the radial differential equation
\begin{equation}
\left( - \frac12 \frac{d^2}{dr^2} - \frac{\d-1}{2 \, r} \frac{d}{dr} + \frac{l
(l+\d-2)}{2 \, r^2} + V(r) \right) \Rad_{n,l}(r) = E_n \Rad_{n,l}(r)
\label{eq:EDO_Rad}
\end{equation}
(see e.g.,~\cite{DehLop10:07,  YanVan99, SanZoz11} for  further details).  Then,
the associated radial probability density $\rho(r)$ is given by
\begin{equation}
\rho_{n,l}(r) dr = \int_{\Sset^{\d-1}} \left| \Psi(\vr) \right|^2 d\vr \, = \,
\left[ \Rad_{n,l}(r) \right]^2 \, r^{\d-1} dr,
\label{eq:RadialPartP}
\end{equation}
where  we have taken  into account  the volume  element $d\vr  = r^{\d-1}  dr \,
d\Omega_{\d-1}$   and  the   normalization  of   the   hyperspherical  harmonics
$\Hyp_{l,\{\mu\}}\left( \Omega_{\d-1} \right)$ to unity.

Then, the wavefunction associated to the  momentum of the system is given by the
Fourier  transform  $\widetilde{\Psi}$ of  $\Psi$.   It  is  known that,  again,
$\widetilde{\Psi}$ writes as the product of a radial and angular part
\begin{equation}
\widetilde{\Psi}_{n,l, \left\lbrace \mu \right\rbrace }(\vk) = \Mad_{n,l}(k) \,
{\cal{Y}}_{l,\{\mu\}}(\Omega_{\d-1}),
\label{eq:PsiTilda_RY}
\end{equation}
with the  the radial part being the modified Hankel transform of $\Rad_{n,l}$,
\begin{equation}
\Mad_{n,l}(k) = (-\imath)^l k^{1-\frac{\d}{2}} \int_{\Rset_+} r^{\frac{\d}{2}}
\Rad_{n,l}(r) \, J_{l+\frac{\d}{2}-1}(k r) \, dr,
\label{eq:Hankel}
\end{equation}
with  $J_\nu$  the Bessel  function  of  the first  king  and  order $\nu$  (see
e.g.,~\cite{Ave02, DehLop10:07,  YanVan99, SanZoz11}).   Again, it leads  to the
radial probability density function
\begin{equation}
\gamma_{n,l}(k) \, = \, \left[ \Mad_{n,l}(k) \right]^2 \, k^{\d-1}.
\label{eq:RadialPartM}
\end{equation}
In  the  following, we  will  focus  on the  $(p,\beta,\lambda)$-Fisher--R\'enyi
complexity of the radial  densities $\rho_{n,l}(r)$ and $\gamma_{n,l}(k)$ of the
$\d$-dimensional harmonic and hydrogenic systems.


\subsection{$(p,\beta,\lambda)$-Fisher--R\'enyi  Complexity  and the  Hydrogenic
  System}

The  bound states  of  a $\d$-dimensional  hydrogenic  system, where  $V(r) =  -
\frac{Z}{r}$ \  ($Z$ denotes the nuclear  charge) are the  physical solutions of
Equation~\ref{eq:EDO_Rad}, which correspond to the known energies
\begin{equation}
E_n^\hyd= -\frac{Z^2}{2\eta^2} \qquad \mbox{where} \qquad \eta = n +
\frac{\d-3}{2}; \qquad n = 1, 2,\ldots
\label{eq:energyH}
\end{equation}
(see~\cite{Nie79, YanVan94, DehLop10:07}).   The radial eigenfunctions are given
by
\begin{equation}
\Rad_{n,l}^\hyd(r) = \sqrt{R_{n,l}} \, {\textstyle \left( \frac{2 Z}{\eta}
\right)^{\!\frac{\d-1}{2}}} \, \tr^{\, l} e^{-\frac{\tr}{2}} \,
\Lag_{\eta-L-1}^{(2 L + 1)}(\tr).
\label{eq:R_nl_H}
\end{equation}
$L$ is  the \textit{grand} orbital angular  momentum quantum number,  $\tr$ is a
dimensionless parameter,  and the normalization coefficient  $R_{n,l}$ are given
by
\begin{equation}
L = l + \frac{\d-3}{2}, \quad l = 0, 1,\ldots, n-1; \qquad \tr = \frac{2 Z}{\eta}
\, r \qquad \mbox{and} \qquad R_{n,l} = \frac{Z \, \Gamma(\eta-L)}{\eta^2 \,
\Gamma(\eta+L+L)},
\label{eq:GrandOrbital_tilder_Rnl_Hydrogen}
\end{equation}
respectively,        with       $\Lag_n^{(\alpha)}(x)$        the       Laguerre
polynomials~\cite{OlvLoz10,  GraRyz07}.  Then,  the  radial probability  density
(\ref{eq:RadialPartP})  of   a  $\d$-dimensional  hydrogenic   stationary  state
$(n,l,\{\mu\})$ is given in position space by
\begin{equation}
\rho_{n,l}^\hyd(r) = R_{n,l} \, \tr^{\, 2 L + 2} \, e^{-\tr} \left[
\Lag_{\eta-L-1}^{(2 L + 1)}(\tr) \right]^2.
\label{eq:RadensityPH}
\end{equation}
Furthermore,    using     8.971    in~\cite{GraRyz07},    one     can    compute
$\frac{d\rho_{n,l}^\hyd}{dr} = \frac{2 Z}{\eta} \frac{d\rho_{n,l}^\hyd}{d\tr}$.

On   the   other  hand,   the   modified   Hankel   transform  of   $\Rad_{n,l}$
Equation~\ref{eq:Hankel}  gives the  radial part  of the  wavefunction  in the
conjugated momentum space as~\cite{Nie79, YanVan94, DehLop10:07}
\begin{equation}
\Mad_{n,l}(k) = \sqrt{M_{n,l}} {\textstyle \left( \frac{\eta}{Z}
\right)^{\frac{\d-1}{2}}} \frac{\tk^{\, l}}{\left( 1 + \tk^{\, 2} \right)^{L+2}}
\: \Geg_{\eta-L-1}^{(L+1)} \left( \frac{1 - \tk^{\, 2}}{1 + \tk^{\, 2}} \right),
\label{eq:M_nl_H}
\end{equation}
where  $\tk$ is  a  dimensionless parameter  and  the normalization  coefficient
$M_{n,l}$ are given by
\begin{equation}
\tk = \frac{\eta}{Z} \, k \qquad \mbox{and} \qquad M_{n,l} = \frac{4^{2L+3} \,
\Gamma(\eta-L) \left[ \Gamma(L+1) \right]^2 \eta^2}{2 \, \pi \, Z \,
\Gamma(\eta+L+1)},
\label{eq:tildek_Mnl_Hydrogen}
\end{equation}
and      where      $\Geg_n^{(\alpha)}(x)$      denotes      the      Gegenbauer
polynomials~\cite{OlvLoz10,  GraRyz07}.   This   gives  the  radial  probability
density function in the momentum space as
\begin{equation}
\gamma_{n,l}^\hyd(k) = M_{n,l} \, \frac{\tk^{\, 2L+2}}{\left( 1 + \tk^{\, 2}
\right)^{2L+4}} \left[ \Geg_{\eta-L+1}^{(L+1)} \left( \frac{1-\tk^{\,
2}}{1+\tk^{\, 2}} \right) \right]^2.
\label{eq:RadensityMH}
\end{equation}
Furthermore,    using     8.939    in~\cite{GraRyz07},    one     can    compute
$\frac{d\gamma_{n,l}^\hyd}{dk}            =           \frac{2           Z}{\eta}
\frac{d\gamma_{n,l}^\hyd}{d\tk}$.

These       expressions       can        thus       be       injected       into
Equations~\ref{eq:RenyiPower}--\ref{eq:triparametricComplexity}  to evaluate
the $(p,\beta,\lambda)$-Fisher--R\'enyi complexity of both $\rho_{n,l}^\hyd$ and
$\gamma_{n,l}^\hyd$.   Due  to  the  special  form  of  the  density,  involving
orthogonal polynomials, this  can be done using for  instance a Gauss-quadrature
method for the integrations~\cite{AbrSte70}.

For  illustration purposes,  we depict  in Figure~\ref{fig:HydrogenicFNCradialR}
the behavior  of the  Fisher information $F_{p,\beta}$,  of the  R\'enyi entropy
power  $N_\lambda$, and  of  the $(p,\beta,\lambda)$-Fisher--R\'enyi  complexity
$C_{p,\beta,\lambda}$  (normalized by the  lower bound)  of the  radial position
density $\rho_{n,l}^\hyd$ of the  $\d$-dimensional hydrogenic system, versus $n$
and $l$, for the parameters  $(p,\beta,\lambda) = (2,1,7)$ and in dimensions $\d
= 3$ and  $12$. Therein, we firstly  observe that, for a given  quantum state of
the system  (so, when $n$ and  $l$ are fixed), the  Fisher information decreases
(see left graph) and the R\'enyi entropy power increases (see center graph) when
$\d$ goes from $3$ to $12$.   This indicates that the oscillatory degree and the
spreading  amount of  the radial  electron  distribution have  a decreasing  and
increasing  behavior,  respectively,  when  the dimension  is  increasing.   The
resulting combined effect, as captured and quantified by the the Fisher--R\'enyi
complexity (see right graph), is such that the complexity has a clear dependence
on the difference  $n-l$ in such a delicate way that  it decreases when $n-l=1,$
but it increases when $n-l$ is bigger than unity as $\d$ is increasing.

To better understand this phenomenon, we  have to look carefully at the opposite
behavior of the Fisher information and the R\'enyi entropy power versus the pair
$(n, l)$.

Indeed, for  the two  dimensionality cases considered  in this work,  the Fisher
information presents  a decreasing  behavior when $l$  is increasing and  $n$ is
fixed,  reflecting essentially  that the  number of  oscillations of  the radial
electron distribution is gradually smaller; keep in mind that $\eta - L = n - l$
is the  degree of  the Laguerre polynomials  which controls the  radial electron
distribution. At the smaller dimension  ($\d=3$), a similar behavior is observed
when $l$ is  fixed and $n$ is increasing, while the  opposite behavior occurs at
the higher dimension ($\d=12$).  This indicates that the radial fluctuations are
bigger in number  as $n$ increases and their amplitudes  depend on the dimension
$\d$ so that they are gradually  smaller (bigger) at the high (small) dimension.
This is  because the dimension, hidden  in both the  hyperquantum numbers $\eta$
and  $L$,  tunes the  coefficients  of the  Laguerre  polynomials  and thus  the
amplitude height of the oscillations.

In the  case of the R\'enyi quantity,  which is a global  spreading measure, the
behavior for  fixed $l$ and $n$  increasing is clearly  increasing, whereas, for
fixed $n$,  it is slowly decreasing  versus $l$; this indicates  that the radial
electron distribution gradually  spreads more and more (less  and less) all over
the space when $n(l)$ is increasing.

\begin{figure}[ht!]
\centerline{\includegraphics[width=\linewidth]{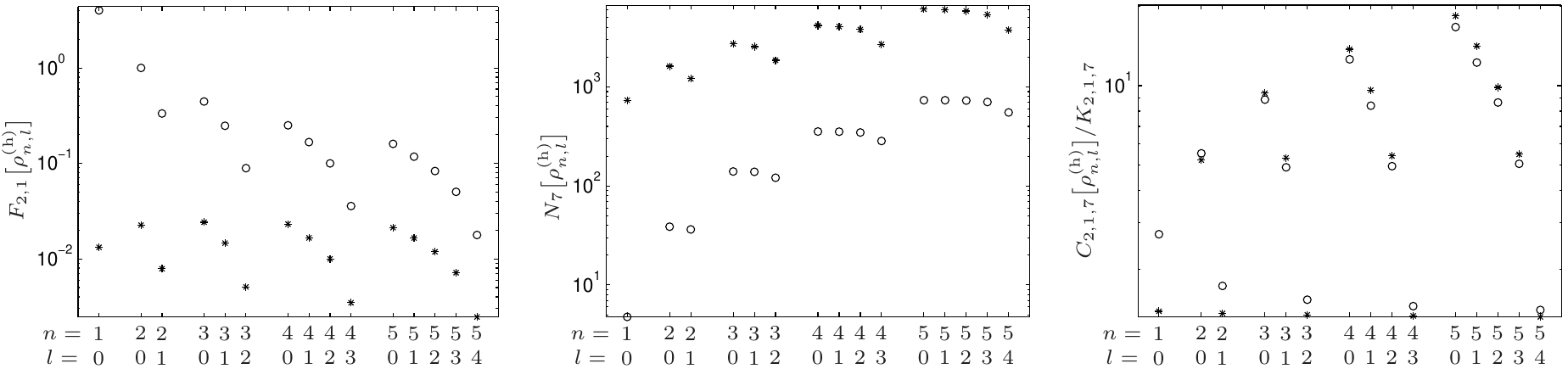}}
\caption{Fisher  information $F_{p,\beta}$ (left  graph), R\'enyi  entropy power
  $N_\lambda$ (center graph), and $(p,\beta,\lambda)$-Fisher--R\'enyi complexity
  $C_{p,\beta,\lambda}$ (right  graph) of the radial  hydrogenic distribution in
  position space with dimensions $\d =  3 (\circ), \: 12 (*)$ versus the quantum
  numbers $n$ and $l$.  The complexity parameters  are $p = 2, \, \beta = 1, \,
  \lambda = 7$.}
\label{fig:HydrogenicFNCradialR}
\end{figure}

Then,  in Figure~\ref{fig:HydrogenicCradialR}, the  parameter dependence  of the
$(p,\beta,\lambda)$-Fisher--R\'enyi   complexity   $C_{p,\beta,\lambda}$   (duly
normalized to  the lower  bound) for the  radial distribution of  various states
$(n,  l)$ of  the  $\d$-dimensional  hydrogenic system  in  position space  with
dimensions $\d = 3$ and $12$,  is investigated for the sets $(p,\beta,\lambda) =
(2,.8,7)$,  \  $(2,1,1)$ (usual  Fisher--Shannon  complexity)  and \  $(5,2,7)$.
Roughly  speaking, the  average behavior  of  the complexity  versus $(n,l)$  is
similar for both  dimensional cases to the  one shown in the right  graph of the
previous  figure. Of  course, for  a  given pair  ($n,l$), the  behavior of  the
complexity in  terms of the  dimension is quantitatively different  according to
the values  of the  parameters. Let us  just point  out, for instance,  that the
comparison of the behavior of $C_{5,2,7}$ versus $\d$ and the corresponding ones
of the other complexities shows that  the complexity with higher value of $p$ is
more sensitive to the radial electron fluctuations with higher amplitudes.

\begin{figure}[ht!]
\centerline{\includegraphics[width=\linewidth]{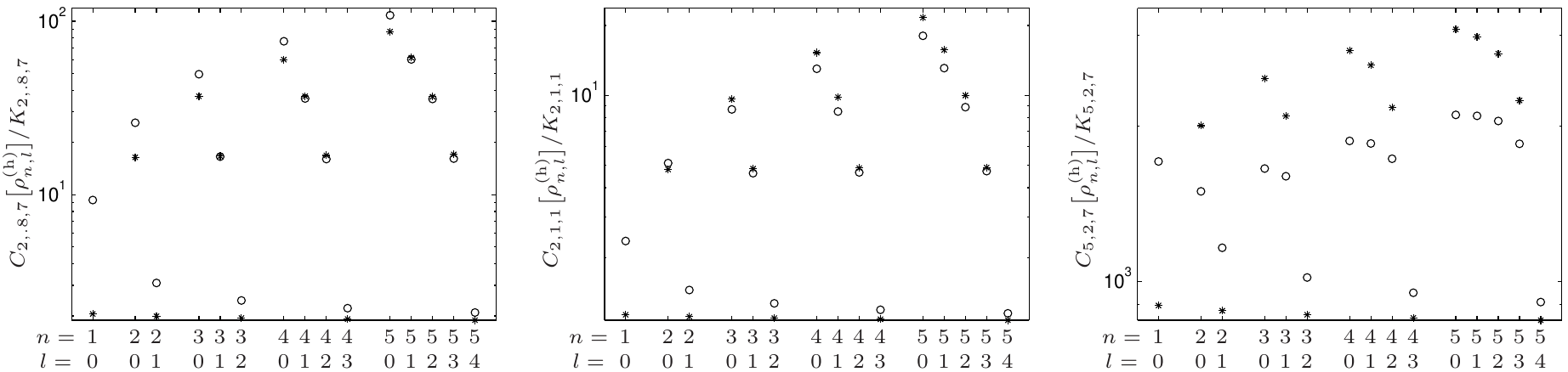}}
\caption{$(p,\beta,\lambda)$-Fisher--R\'enyi complexity (normalized to its lower
  bound), $C_{p,\beta,\lambda}$, with  $ (p,\lambda,\beta) =~(2,0.8,7), (2,1,1),
  (5,2,7)$ for  the radial  hydrogenic distribution in  the position  space with
  dimensions $\d = 3 (\circ)$ and $12 (*) $.}
\label{fig:HydrogenicCradialR}
\end{figure}
 
A  similar study  for  the  previous entropy-  and  complexity-like measures  in
momentum   space  has   been   done  in   Figures~\ref{fig:HydrogenicFNCradialM}
and~\ref{fig:HydrogenicCradialM}. Briefly, we observe that the behavior of these
momentum quantities  are in  accordance with the  analysis of  the corresponding
ones in position space, which has  just been discussed. Note that here again the
difference  $n-l$  determines the  degree  of  the  Gegenbauer polynomials  that
control the momentum density $\gamma_{n,l}^\hyd$, so that the influence of $n,l$
and $\d$ is formally similar to that for the position density $\rho_{n,l}^\hyd$.
Here,  the influence of  $\d$ on  the height  of the  radial oscillation  of the
electron distribution  (through the coefficients of  the Gegenbauer polynomials)
is the same for the two dimensionality cases considered in this work.

Let us  highlight that the $(n,l,\d)$-behavior  of the R\'enyi  power entropy in
momentum  space  is  just  the  opposite  to  the  corresponding  position  one,
manifesting the conjugacy of the two spaces, which is the spread of the position
and momentum electron distributions are opposite.

\begin{figure}[ht!]
\centerline{\includegraphics[width=\linewidth]{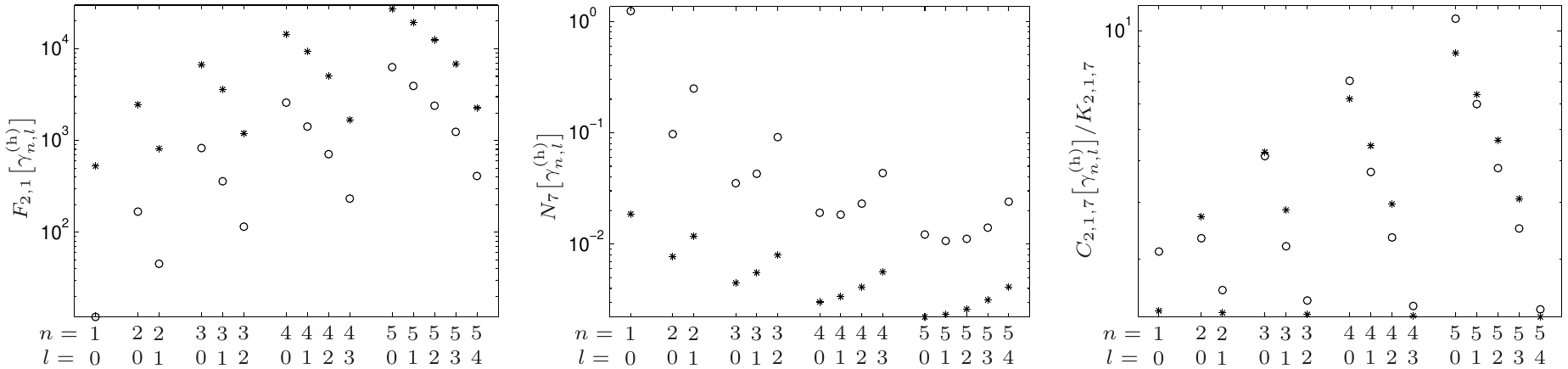}}
\caption{Fisher  information $F_{p,\beta}$ (left  graph), R\'enyi  entropy power
  $N_\lambda$ (center graph), and $(p,\beta,\lambda)$-Fisher--R\'enyi complexity
  $C_{p,\beta,\lambda}$ (right  graph) of the radial  hydrogenic distribution in
  momentum space with dimensions $\d = 3 (\circ), \: 12 (*) $ versus the quantum
  numbers $n$ and $l$.   The complexity parameters are $p = 2,  \, \beta = 1, \,
  \lambda = 7$.}
\label{fig:HydrogenicFNCradialM}
\end{figure}

\begin{figure}[ht!]
\centerline{\includegraphics[width=\linewidth]{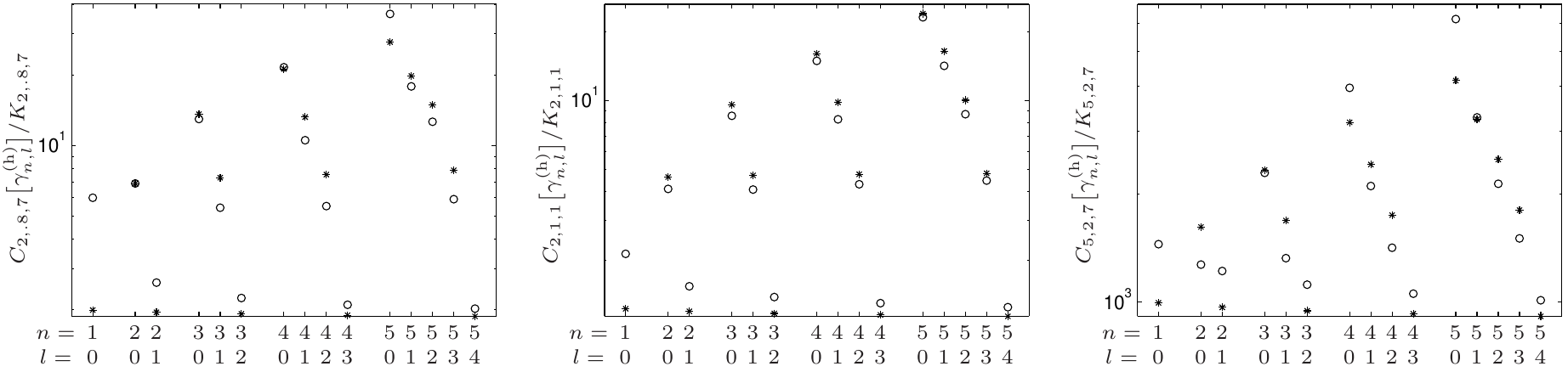}}
\caption{$(p,\beta,\lambda)$-Fisher--R\'enyi complexity (normalized to its lower
  bound),  $C_{p,\beta,\lambda}$,   with  $(p,\lambda,\beta)  =   (2,0.8,7),  \:
  (2,1,1), \:  (5,2,7)$ for the  radial hydrogenic distribution in  the momentum
  space with dimensions $\d = 3 (\circ)$ and $12 (*)$.}
\label{fig:HydrogenicCradialM}
\end{figure}


\subsection{$(p,\beta,\lambda)$-Fisher--R\'enyi  Complexity   and  the  Harmonic
  System}

The bound states of  a $\d$-dimensional harmonic (i.e., oscillator-like) system,
where $V(r) = \frac12 \omega^2 \,  r^2$ (without loss of generality, the mass is
assumed  to   be  unity),  are  known  to   have  the energies
\begin{equation}
E^\osc_n = \omega \left( 2 n + L + \frac{3}{2} \right) \qquad \mbox{with}
\qquad n = 0, 1,\ldots, \quad l = 0, 1,\ldots
 \label{eq:energyO}
\end{equation}
(see    e.g.,~\cite{LouSha60:p1,    LouSha60:p2,    YanVan94}).    The    radial
eigenfunctions writes in terms of the Laguerre polynomials as
\begin{equation}
\Rad_{n,l}^\osc(\tr) = \sqrt{R_{n,l}} \: \omega^{\frac{\d-1}{4}} \, \tr^{\, l} \,
e^{-\frac12 \tr^{\, 2}} \, \Lag_n^{(L+\frac12)} \left( \tr^{\, 2} \right),
\label{eq:R_nl_O}
\end{equation}
where  $\tr$ is  a dimensionless  parameter, and  the  normalization coefficient
$R_{n,l}$ are given by
\begin{equation}
\tr = \sqrt{\omega} \, r \qquad \mbox{and} \qquad R_{n,l} = \frac{2
\sqrt{\omega} \, \Gamma(n+1)}{\Gamma\left( n+L+\frac32 \right)},
\label{eq:tilder_Rnl_Oscillator}
\end{equation}
respectively. Then, the associated radial position density is thus given by
\begin{equation}
\rho_{n,l}^\osc(r) = R_{n,l} \, \tr^{\, 2L+2} \, e^{- \tr^{\, 2}} \, \left[
\Lag_n^{(L+\frac12)} \left( \tr^{\, 2} \right) \right]^2.
\label{eq:RadensityPO}
\end{equation}
As for  the hydrogenic system,  using 8.971 in~\cite{GraRyz07}, one  can compute
$\frac{d\rho_{n,l}^\osc}{dr} = \sqrt{\omega} \frac{d\rho_{n,l}^\osc}{d\tr}$, and
thus the  $(p,\beta,\lambda)$-Fisher--R\'enyi of $\rho_{n,l}^\osc$.  Remarkably,
$\Rad_{n,l}$ is invariant by the modified Hankel transform, so that the momentum
radial density is formally the same as the position radial density.

For  illustration purposes, we  plot in  Figure~\ref{fig:HarmonicFNCradialR} the
behavior of  the Fisher information  $F_{2,1}$, the R\'enyi entropy  power $N_7$
and the $(2,1,7)$-Fisher--R\'enyi complexity  $C_{2,1,7}$ of the radial position
distribution of the  $\d$-dimensional harmonic system for various  values of the
quantum  numbers   $n$  and  $l$   at  the  dimensions   $\d  =  3$   and  $12$.
Figure~\ref{fig:HarmonicCradialR}     depicts     $C_{p,\beta,\lambda}$     duly
renormalized  by its  lower bound,  for  the triplets  of complexity  parameters
$(p,\beta,\lambda) = (2,.8,7), \: (2,1,1)$ and $(5,2,7)$, respectively. In these
graphs, one can make a similar  interpretation as for the hydrogenic case. Note,
however,  that here  the  degree of  the  Laguerre polynomials  involved in  the
distribution $\rho_{n,l}^\osc$ only depends on $n$; this fact makes more regular
the behavior of the  previous information-theoretical measures in the oscillator
case than in the hydrogenic one.  Concomitantly, as $n$ increases, the spreading
of the distribution  also increases. Conversely, parameters $l$  and $\d$ have a
relatively small influence on both the  smoothness of the oscillation and on the
spreading  (compared to  that of  $n$).  Thus,  unsurprisingly, both  the Fisher
information  and  the  R\'enyi  entropy  power  are  weakly  influenced  by  $l$
(especially  at  the  higher   dimension)  and  by  $\d$.   The  Fisher--R\'enyi
complexity,  which quantifies  the combined  oscillatory and  spreading effects,
exhibits a very regular increasing behavior in terms of $n$.

Most interesting is the parameter-dependence  of the complexity.  Indeed, we can
play with the complexity parameter to stress different aspects of the oscillator
density  and  thus to  reveal  differences between  the  quantum  states of  the
system. For  instance, as one can see  in Figure~\ref{fig:HarmonicCradialR}, the
usual Fisher--R\'enyi  complexity is unable  to quantify the  difference between
the states of a  given $n$ versus the orbital number $l$  and the dimension $\d$
(especially  when $n \ge  1$, whereas  the systems  are quite  different).  This
holds even  playing with  $\lambda$ or $\beta$,  while increasing  parameter $p$
(right graph), these  states are distinguishable.  This graph  clearly shows the
potentiality of  the family of  complexities $C_{p,\beta,\lambda}$ to  analyze a
system, especially  thanks to  the full  degree of freedom  we have  between the
complexity parameters $p, \, \beta$ and $\lambda$.

\begin{figure}[ht!]
\centerline{\includegraphics[width=\linewidth]{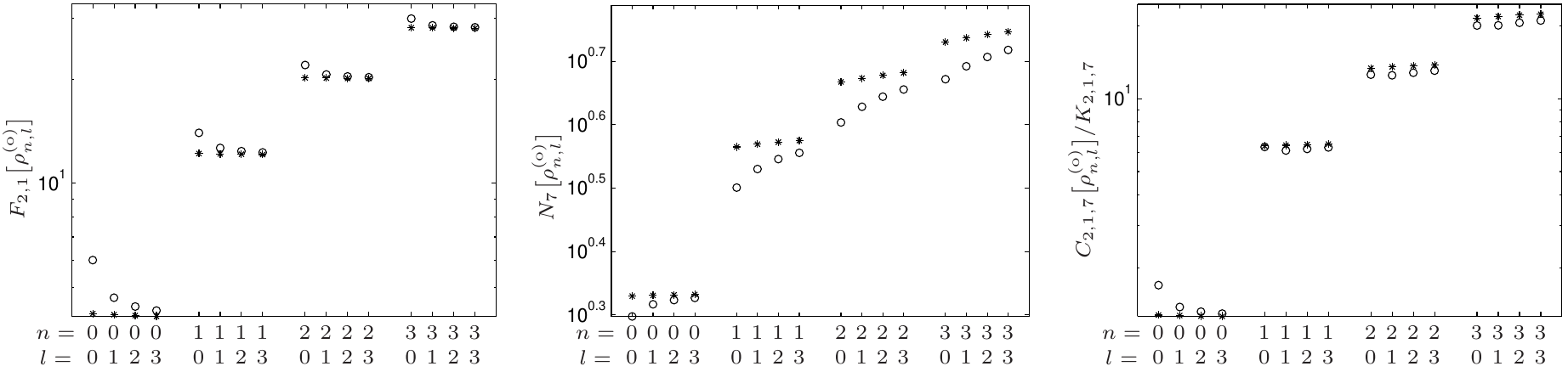}}
\caption{Fisher  information $F_{p,\beta}$ (left  graph), R\'enyi  entropy power
  $N_\lambda$ (center graph), and $(p,\beta,\lambda)$-Fisher--R\'enyi complexity
  $C_{p,\beta,\lambda}$ (right graph) versus $n$ and $l$ for the radial harmonic
  system in  position space with  dimensions $\d  = 3 (\circ),  \ 12 (*)  $. The
  informational parameters are $p = 2, \, \beta = 1, \, \lambda =7$.}
\label{fig:HarmonicFNCradialR}
\end{figure}

\begin{figure}[ht!]
\centerline{\includegraphics[width=\linewidth]{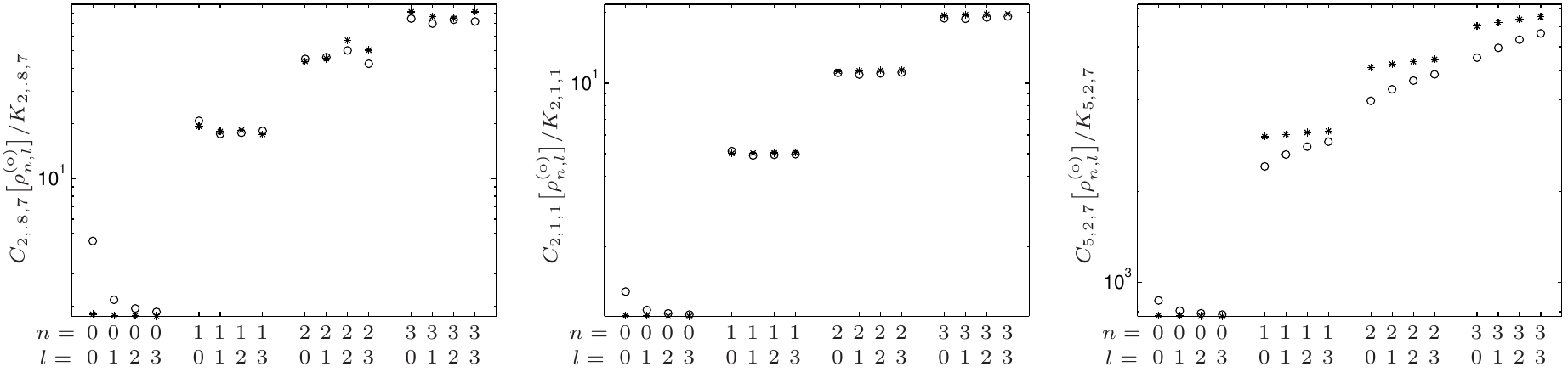}}
\caption{$(p,\beta,\lambda)$-Fisher--R\'enyi complexity (normalized to its lower
  bound)   $C_{p,\beta,\lambda}$    with   $(p,\lambda,\beta)   =   (2,0.8,7),\,
  (2,1,1),\,(5,2,7)$  for  the oscillator  system  in  the  position space  with
  dimensions $\d = 3 (\circ), 12 (*) $.}
\label{fig:HarmonicCradialR}
\end{figure}


\section{Conclusions}
\label{sec:Conclusion}

In  this  paper,  we  have  defined a  three-parametric  complexity  measure  of
Fisher--R\'enyi  type   for  a   univariate  probability  density   $\rho$  that
generalizes all the previously published  quantifiers of the combined balance of
the spreading and oscillatory facets of  $\rho$. We have shown that this measure
satisfies the three fundamental  properties of a statistical complexity, namely,
the invariance  under translation and scaling transformations  and the universal
bounding from  below.  Moreover,  the minimizing distributions  are found  to be
closely  related to  the  stretched  Gaussian distributions.   We  have used  an
approach    based    on   the    Gagliardo--Nirenberg    inequality   and    the
differential-escort  transformation  of $\rho$.  In  fact,  this inequality  was
previously used by  Bercher and Lutwak et al.  to  find a biparametric extension
of the celebrated  Stam inequality which lowerbounds the  product of the R\'enyi
entropy power  and the Fisher  information.  We have extended  this biparametric
Stam   inequality   to  a   three-parametric   one   by   using  the   idea   of
differential-escort deformation of a probability density.

Then, we have numerically analyzed  the previous entropy-like quantities and the
three-parametric complexity  measure for various specific quantum  states of the
two main prototypes of multidimensional  electronic systems subject to a central
potential of  Coulomb (the $\d$-dimensional  hydrogenic atom) and  harmonic (the
$\d$-dimensional  isotropic harmonic  oscillator) character.   Briefly,  we have
found that the  proposed complexity allows to capture  and quantify the delicate
balance  of the  gradient  and the  spreading  contents of  the radial  electron
distribution of  ground and excited states  of the system. The  variation of the
three parameters  of the  proposed complexity allows  one to  stress differently
this balance in the various radial regions of the charge distribution.

The results  found in this work can  be generalized in various  ways that remain
open.  Indeed, the  Gagliardo--Nirenberg  relation is  quite  powerful since  it
involves the $p$-norm of the function $u$, the $q$-norm of its $j$-th derivative
and the $s$-norm of its $m$-th derivative,  where $p, q, s$ and the integers $j,
m$  are  linked  by   inequalities  (see~\cite{Nir59}).  This  leaves  open  the
possibility to  define still more extended (complete)  complexity measures, with
higher-order  (in  terms of  derivative)  measures  of  information.  Even  more
interesting, this inequality-based relation holds  for any dimension $\d \ge 1$;
thus,  it  supports  the  possibility   to  extend  our  univariate  results  to
multidimensional distributions, but with tighter restrictions on the parameters.
The main difficulty in this  case is related with the multidimensional extension
of the validity  domain by using the differential-escort  technique or a similar
one.


\vspace{6pt}

\textbf{Acknowledgments } {The authors  are very grateful  to the CNRS (Steeve  Zozor) and
  the Junta de Andaluc\'ia and the MINECO--FEDER under the grants FIS2014--54497
  and FIS2014--59311P  (J.S.D.) for  partial financial support.   Moreover, they
  are grateful for the warm hospitality  during their stays at GIPSA--Lab of the
  University   of  Grenoble--Alpes  (D.P.C.)    and  Departamento   de  F\'isica
  At\'omica, Molecular y Nuclear of the University of Granada (S.Z.)  where this
  work was partially carried out.}

\textbf{authorcontributions}{The authors contributed equally to this work.}



\appendix


\section{Proof of Proposition~\ref{prop:StamDp}}
\label{app:StamDp}


\subsection{The Case $\lambda \ne 1$}

The result of the proposition is a direct consequence of the Gagliardo--Nirenberg
\linebreak inequality~\cite{Nir59, Agu08, Ber12:06_1},  stated in our context as
follows: let $p  > 1, \, s > q \ge  1$ and $\theta = \frac{p (s-q)}{s  (p + pq -
  q)}$; then, there exists an  optimal strictly positive constant $K$, depending
only on $p, q$ and $s$ such that for any function $u: \Rset \mapsto \Rset_+$,
\begin{equation}
K \left\| \frac{d}{dx} u \right\|_p^\theta \left\| u \right\|_q^{1-\theta} \,
\ge \, \left\| u \right\|_s,
\label{eq:GagliardoNirenberg}
\end{equation}
provided that the involved quantities exist, the equality being achieved for $u$
solution of the differential equation
\begin{equation}
- \frac{d}{dx} \left( \left| \frac{d}{dx} u \right|^{p-2} \frac{d}{dx} u
\right) + u^{q-1} = \gamma u^{s-1},
\end{equation}
where $\gamma > 0$ is such that $\|u\|_s$ is fixed and can be chosen arbitrarily
(it  corresponds  to  a  Lagrange   multiplier,  see  Equations  (26)  and  (27)
in~\cite{Agu08}).   Finding $u$ thus  allows to  determine the  optimal constant
$K$.  Note that, if the equality in~\ref{eq:GagliardoNirenberg} is reached for
$u_\gamma$,  then   it  is  also  reached  for   $\overline{u}_\gamma  =  \delta
u_\gamma(x)$  for any  $\delta >  0$.   One can  see that  $\overline{u}_\gamma$
satisfies the differential equation  $ - \frac{d}{dx} \left( \left| \frac{d}{dx}
    \overline{u} \right|^{p-2} \frac{d}{dx}  \overline{u} \right) + \delta^{p-q}
\, \overline{u}^{q-1}  - \gamma \delta^{p-s} \, \overline{u}^{s-1}  = 0$.  Thus,
function $u$ reaching the equality in Equation~\ref{eq:GagliardoNirenberg} can
also  be chosen as  the solution  of the  differential equation  $- \frac{d}{dx}
\left(  \left| \frac{d}{dx}  u  \right|^{p-2} \frac{d}{dx}  u  \right) +  \kappa
u^{q-1} - \zeta u^{s-1}  = 0$, where $\kappa > 0$ and $\zeta  > 0$ can be chosen
arbitrarily. As we will see later on, a judicious choice allowing to include the
limit case $s \to q$ is to  take $\kappa = \zeta = \frac{\gamma}{s-q}$, i.e., to
chose       function       $u$        reaching       the       equality       in
Equation~\ref{eq:GagliardoNirenberg}  as  the  solution  of  the  differential
equation
\begin{equation}
- \frac{d}{dx} \left( \left| \frac{d}{dx} u \right|^{p-2} \frac{d}{dx} u \right)
+ \gamma \: \frac{u^{q-1} - u^{s-1}}{s-q} = 0,
\label{eq:EDOGN}
\end{equation}
where $\gamma > 0$ can be arbitrarily chosen.


\subsubsection{The Sub-Case $\lambda < 1$}

Following the very same steps than in~\cite{Ber12:06_1}, let us consider first
$$\lambda = \frac{q}{s} < 1.$$  With  $u^s$ integrable, one can normalize it, that
is, writing it under the form $u = \rho^{\frac{1}{s}} = \rho^{\frac{\lambda}{q}}$
with $\rho$  a probability density  function.  Thus, $\|u\|_s  = 1$ and  from the
Gagliardo--Nirenberg     inequality,     $$\left\|     \rho^{\frac{\lambda}{q}-1}
  \frac{d}{dx}   \rho  \right\|_p^\theta  \,   \left\|  \rho^{\frac{\lambda}{q}}
\right\|_q^{1-\theta} \, \ge \, s^\theta K^{-1}.$$ Simple algebra allows to write
the terms  of the left-hand side in terms of the  generalized Fisher information
and of the R\'enyi entropy power, respectively, to conclude that
\begin{equation}
\Big( F_{p,\frac{\lambda}{q} - \frac{1}{p} + 1}[\rho] \Big)^{\frac{\theta}{p} \,
\frac{p \left( \frac{\lambda}{q} - \frac{1}{p} + 1 \right)}{2}} \Big(
N_\lambda[\rho] \Big)^{\frac{1-\theta}{q} \, \frac{1-\lambda}{2}} \ge s^\theta K^{-1}.
\label{eq:ProvisoryStamInf}
\end{equation}
Using   $1-\frac{1}{p}  =  \frac{1}{p^*}$,   let  us   then  denote   $$\beta  =
\frac{\lambda}{q} + \frac{1}{p^*} = \frac{1}{s} + \frac{1}{p^*},$$ and note that
the  conditions  imposed  on  $p,  q$  and  $s$  together  with  $\lambda  >  0$
impose  $$\beta  \in  \left(  \frac{1}{p^*}  \, ;  \,  \frac{1}{p^*}  +  \lambda
\right],$$ once $p$ and $\lambda$ are given.  Simple algebra allows thus to show
that  $\frac{\theta}{p}  \,  \frac{p  \left( \frac{\lambda}{q}  +  \frac{1}{p^*}
  \right)}{2}  =   \frac{1-\theta}{q}  \,  \frac{1-\lambda}{2}   =  \frac{\theta
  \beta}{2} >  0$: the  exponent of  the Fisher information  and of  the entropy
power in  Equation~\ref{eq:ProvisoryStamInf} are thus  equal.  Moreover, $\theta$
being strictly  positive, both sides of Equation~\ref{eq:ProvisoryStamInf} can be
elevated  to   exponent  $\frac{2}{\theta}$  leading   to  the  result   of  the
proposition, where the bound is given by
\begin{equation}
K_{p,\beta,\lambda} = s^2 K^{-\frac{2}{\theta}},
\label{eq:K_lambda_inf_1}
\end{equation}
where $s$ and  $\theta$ can be expressed by their  parametrization in $p, \beta,
\lambda$.  Finally, the differential equation~\ref{eq:EDOMin} satisfied by the
minimizer $u$ comes from  Equation~\ref{eq:EDOGN} noting that $s = \frac{p^*}{\beta
  p^* - 1}$ and $q =  \frac{\lambda p^*}{\beta p^* - 1}$, remembering that $\rho
= u^s$ and thus that $\gamma$ is to be chosen such that $u^s$ sums to unity.


\subsubsection{The Sub-Case $\lambda > 1$}

Consider  now  $$\lambda =  \frac{s}{q}  > 1$$  and  $u  = \rho^{\frac{1}{q}}  =
\rho^{\frac{\lambda}{s}}$, leading to
\begin{equation}
\Big( F_{p,\frac{\lambda}{s} - \frac{1}{p} + 1}[\rho] \Big)^{\frac{\theta}{p} \,
\frac{p \left( \frac{\lambda}{s} - \frac{1}{p} + 1 \right)}{2}} \Big(
N_\lambda[\rho] \Big)^{- \frac{1}{s} \, \frac{1-\lambda}{2}} \ge q^\theta K^{-1}.
\label{eq:ProvisoryStamSup}
\end{equation}
Denoting  now  $$\beta  =  \frac{\lambda}{s}  + \frac{1}{p^*}  =  \frac{1}{q}  +
\frac{1}{p^*},$$ imposing $$\beta \in \left( \frac{1}{p^*} \, ; \, \frac{1}{p^*}
  + 1 \right]$$  once $p$ and $\lambda$ are given.   Simple algebras allows thus
to show that $\frac{\theta}{p} \, \frac{p \left( \frac{\lambda}{s} - \frac{1}{p}
    +  1  \right)}{2} =  -  \frac{1}{s}  \,  \frac{1-\lambda}{2} =  \frac{\theta
  \beta}{2} >  0 $:  again, the exponent  of the  Fisher information and  of the
entropy  power  in  Equation~\ref{eq:ProvisoryStamSup}  are equal.   Here  again,
$\theta > 0$ allowing  to elevate both side of Equation~\ref{eq:ProvisoryStamSup}
to exponent $\frac{2}{\theta}$. The bound is now given by
\begin{equation}
K_{p,\beta,\lambda} = q^2 K^{-\frac{2}{\theta}}
\label{eq:K_lambda_sup_1}
\end{equation}
where $q$ and  $\theta$ can be expressed by their  parametrization in $p, \beta,
\lambda$.    Finally,    as   for   the   previous    case,   the   differential
equation~\ref{eq:EDOMin}   satisfied   by  the   minimizer   $u$  comes   from
Equation~\ref{eq:EDOGN}  noting that  now $q  = \frac{p^*}{\beta  p^*-1}$ and  $s =
\frac{\lambda p^*}{\beta  p^*-1}$, remembering  that now $\rho  = u^q$  and thus
that $\gamma$ is to be chosen such that $u^q$ sums to unity.


\subsection{The Case $\lambda = 1$}

The minimizer for  $\lambda = 1$ can  be viewed as the limiting  case $\lambda \to
1$,  i.e., $s  \to  q$.   

One can  also process as  done by Agueh  in~\cite{Agu08} to determine  the sharp
bound of  the Gagliardo--Nirenberg inequality. To  this end, let  us consider the
minimization problem
\begin{equation}
\inf \left\{ \frac{1}{p} \int_{\Rset} \Big| \frac{d}{dx} u(x) \Big|^p dx -
\frac{1}{q} \int_{\Rset} [u(x)]^q \log u(x) \, dx : \quad
u \ge 0, \quad \int_{\Rset} [u(x)]^q \, dx = 1\right\}
\label{eq:inf_problem}
\end{equation}
for $p  > 1$  and $q  \ge 1$ (see Chapters~5  and 6 in~\cite[]{GelFom63}  justifying the
existence of a  minimum). Hence, there exists an optimal  constant $K$ such that
for any function $u$ such that $u^q$ sums to unity,
\begin{equation}
\frac{1}{p} \int_{\Rset} \left| \frac{d}{dx} u(x) \right|^p \, dx - \frac{1}{q}
\int_{\Rset} [u(x)]^q \log u(x) \, dx \ge K.
\label{eq:inequality_inf_problem}
\end{equation}
Now, fix a  function $u$ and consider $v(x)  = \gamma^{\frac{1}{q}} u(\gamma x)$
for some  $\gamma >  0$. $v^q$ also  sums to unity  and thus  can be put  in the
previous inequality, leading to
\begin{equation}
f_u(\gamma) \equiv \frac{\gamma^{\frac{p}{q}+p-1}}{p} \int_{\Rset} \left|
\frac{d}{dx} u(x) \right|^p \, dx - \frac{1}{q} \int_{\Rset} [u(x)]^q \log u(x)
\, dx - \frac{1}{q^2} \log \gamma \ge K
\label{eq:min_f_gamma}
\end{equation}
for any  $\gamma > 0$.  Thus,  this inequality is necessarily  satisfied for the
$\gamma$  that  minimizes $f_u(\gamma)$.   A  rapid  study  of $f_u$  allows  to
conclude that it is minimum for
\begin{equation}
\gamma = \left( \frac{p}{\displaystyle q ( p + q (p-1) ) \int_{\Rset} \Big|
\frac{d}{dx} u(x) \Big|^p dx} \right)^{\frac{q}{p+q(p-1)}}.
\label{eq:gamma_opt}
\end{equation}
Now, injecting Equation~\ref{eq:gamma_opt} in Equation~\ref{eq:min_f_gamma} gives
\begin{equation}
\frac{1}{p+q(p-1)} \log \int_{\Rset} \Big| \frac{d}{dx} u(x) \Big|^p dx -
\int_{\Rset} [u(x)]^q \log u(x) \, dx \ge \widetilde{K},
\label{eq:modified_log_Sobolev}
\end{equation}
with $\widetilde{K} = q K + \frac{1}{p + q (p-1)} \left( \log\left( \frac{p}{q (
      p + q (p-1) )} \right) - 1 \right)$. Consider now $u_{\min}$ the minimizer
of problem~\ref{eq:inf_problem}.  Obviously, $f_{u_{\min}}(\gamma)$ is minimum
for  $\gamma = 1$,  that gives,  from Equation~\ref{eq:gamma_opt},  \ $\int_{\Rset}
\Big| \frac{d}{dx}  u_{\min}(x) \Big|^p dx = \frac{p}{q  (p + q (p-1)  )}$ \ and
from Equation~\ref{eq:inequality_inf_problem}, being an equality, \ $\int_{\Rset}
[u_{\min}(x)|^q \log  u_{\min}(x) \,  dx =  \frac{1}{p + q  (p-1) }  - q  K$.  \
Injecting these  expressions in Equation~\ref{eq:modified_log_Sobolev}  allows
concluding  that  this  inequality  is  sharp, and  moreover  that  its  minimizer
coincides with that of the minimization problem~\ref{eq:inf_problem}.

Inequality~\ref{eq:StamDp} is  obtained by injecting  $u = \rho^{\frac{1}{q}}$
in  Equation~\ref{eq:modified_log_Sobolev} and  after some  trivial  algebra and
denoting $\beta = \frac{1}{q} + \frac{1}{p^*} \in \left( \frac{1}{p^*} \, ; \, 1
  + \frac{1}{p^*}  \right]$, confirming that  it can be  viewed as a  limit case
$\lambda \to 1$.

Let us now solve  the minimization problem~\ref{eq:inf_problem}, that is, from
the Lagrangian technique~\cite{Bru04}, to minimize $\int_{\Rset} F(x,u,u') dx,$ \
where \ $F(x,u,u')  = \frac{1}{p} \Big| \frac{d}{dx} u(x)  \Big|^p - \frac{1}{q}
[u(x)]^q \log  u(x) \, -  \gamma [u(x)]^q$ and  where $u' = \frac{d}{dx}  u$ and
$\gamma$ is the Lagrange multiplier. The solution of this variational problem is
given by  the Euler--Lagrange equation~\cite{Bru04},  $\frac{\partial F}{\partial
  u} -  \frac{d}{dx} \left(  \frac{\partial F}{\partial u'}  \right) =  0$, that
writes here after a re-parametrization $\delta = \frac{1}{q} + q \gamma$
\begin{equation}
  - \frac{d}{dx} \left( \left| \frac{d}{dx} u \right|^{p-2} \frac{d}{dx} u
\right) - u^{q - 1} \left( \log u + \delta \right) = 0.
\end{equation}
$\delta$  is to  be determined  a  posteriori so  as to  satisfy the  constraint
$\int_\Rset [u(x)]^q dx  = 1$.  Again, one  can easily see that if  the bound in
Equation~\ref{eq:modified_log_Sobolev}  is achieved  for $u_{\min}$,  then  it is
also  achieved  for   $\overline{u}_\delta(x)  =  \sigma  u_{\min}(\sigma^q  x)$
whatever    $\sigma    >    0$.     Reporting   $u_{\min}(x)    =    \sigma^{-1}
\overline{u}_\sigma(\sigma^{-q} x)$  in the differential equation  allows to see
that  $\overline{u}_\sigma$  is a solution of  the  differential  equation  $  -
\frac{d}{dx} \left( \left|  \frac{d}{dx} \overline{u} \right|^{p-2} \frac{d}{dx}
  \overline{u}  \right) -  \sigma^{p+q(p-1)}  \overline{u}^{q -  1} \left(  \log
  \overline{u}  -  \log  \sigma +  \delta  \right)  =  0$.  Choosing  $\sigma  =
\exp(\delta)$ and  rewriting $\sigma^{p+q(p-1)} =  \gamma$, one can  thus choose
the minimizer $u$ as the solution of the differential equation
\begin{equation}
  - \frac{d}{dx} \left( \left| \frac{d}{dx} u \right|^{p-2} \frac{d}{dx} u
\right) - \gamma \, u^{q - 1} \log u = 0,
\label{eq:EDOLogSoboloev}
\end{equation}
where $\gamma$ is to be determined  a posteriori so as to satisfy the constraint
$\int_\Rset [u(x)]^q  dx = 1$. This result is precisely the limit case of the
differential equation~\ref{eq:EDOGN} when $s \to q$.


\section{Proof of Proposition~\ref{prop:Symmetries}}
\label{app:Symmetries}

For      $\lambda      =      1,$     Relations~\ref{eq:RhoInvolution}
and~\ref{eq:KInvolution} induced by Transform~\ref{eq:Involution} of the
indexes are obvious since $\opT_p(\beta,1) = (\beta,1)$.

Then, for $\lambda \ne  1$, Relation~\ref{eq:RhoInvolution} comes from the
fact that the function $u$  solution of Equation~\ref{eq:EDOGN} depends only on $p,
q$ and $s$.   Let us write $(\beta,\lambda)$ and  $\vartheta$ the parameters for
the first situation of the above proof, i.e., $\lambda = \frac{q}{s}$ and $\beta
=    \frac{1}{s}+\frac{1}{p^*}     =    \frac{\lambda}{q}+\frac{1}{p^*}$,    and
$(\overline{\beta},\overline{\lambda})$     and    $\overline{\vartheta}$    the
parameters for  the second  situation, i.e., $\overline{\lambda}  = \frac{s}{q}$
and        $\overline{\beta}         =        \frac{1}{q}+\frac{1}{p^*}        =
\frac{\overline{\lambda}}{s}+\frac{1}{p^*}$.  It is  straightforward to see that
$\overline{\lambda}    =     \frac{1}{\lambda}$    and    $\overline{\beta}    =
\frac{\beta}{\lambda} - \frac{1}{\lambda p^*}  + \frac{1}{p^*} = \frac{\beta p^*
  + \lambda - 1}{\lambda p^*}$, i.e., $(p,\overline{\beta},\overline{\lambda}) =
(p,\opT_p(\beta,\lambda))$,   and,   conversely,   that   $(p,\beta,\lambda)   =
(p,\opT_p(\overline{\beta},\overline{\lambda}))$.   Since  the  optimal  $u$  is
fixed once  $p, q$  and $s$ are  given, one has  $u_{p,\opT_p(\beta,\lambda)} =
u_{p,\beta,\lambda}$.    Finally,   simple    algebra   allows   to   show   that
$\overline{\vartheta}    =   \lambda    \,   \vartheta$    and    $\vartheta   =
\overline{\lambda} \, \overline{\vartheta}$, which finishes the proof.

Now,      Relation~\ref{eq:KInvolution}       immediately      comes      from
Equations~\ref{eq:K_lambda_inf_1}   and~\ref{eq:K_lambda_sup_1}   together   with
$\lambda = \frac{q}{s}$.


\section{Proof of Proposition~\ref{prop:StamDpTilda}}
\label{app:StamDpTilda}


\subsection{The  $(p,\beta,\lambda)$-Fisher--R\'enyi Complexity  is Lowerbounded
  over $\widetilde{\Dset}_p$}\label{app:StamDpTildaExistence}

As   detailed   in   the   text,   consider   a   point   $(\beta,\lambda)   \in
\widetilde{\Dset}_p$.   Thus, there  exists  an  index $\alpha  >  0$ such  that
$\opA_\alpha(\beta,\lambda)  \in  \Lset_p  \cup  \overline{\Lset}_p$.   Applying
Propositions~\ref{prop:StamDp}                                                 and ~\ref{prop:ComplexityStretchedEscort}, we have
\begin{eqnarray*}
C_{p,\beta,\lambda}[\rho] & = & \alpha^2 \, C_{p,\opA_\alpha(\beta,\lambda)}\Big[
\opE_\alpha\big[ \opE_{\alpha^{-1}}[\rho] \big] \Big]\\[4mm]
& \ge & \alpha^2 \, K_{p,\opA_\alpha(\beta,\lambda)} \: \equiv \:
K_{p,\beta,\lambda}.
\end{eqnarray*}

Finally,     denoting    $(\widetilde{\beta}     ,     \widetilde{\lambda})    =
\opA_\alpha(\beta,\lambda)$,  the  minimizers  satisfy  $\opE_{\alpha^{-1}}\big[
\rho_{p,\beta,\lambda}      \big]     =      g_{p,\widetilde{\lambda}}$     (see
Section~\ref{sec:Bercher}),   or   $\opE_{\alpha^{-1}}\big[  \rho_{p,\beta,\lambda}
\big] = g_{p,2-\widetilde{\lambda}}$ (see Section~\ref{sec:BercherSymmetric}), that
is,
\begin{eqnarray*}
\rho_{p,\beta,\lambda} = \left\{\begin{array}{lll}
\opE_{\alpha}\big[ g_{p,\widetilde{\lambda}} \big], & \mbox{if} &
\opA_\alpha(\beta,\lambda) \in \Lset_p,\\[4mm]
\opE_{\alpha}\big[ g_{p,2-\widetilde{\lambda}} \big], & \mbox{if} &
\opA_\alpha(\beta,\lambda) \in \overline{\Lset}_p.
\end{array}\right.
\end{eqnarray*}


\subsection{Explicit Expression for the Minimizers.}

In   the   sequel,   we   determine   the   differential-escort   transformation
$\opE_\alpha[g_{p,\lambda}]$   with   $\lambda  <   1$.    Let   us  denote   by
$\displaystyle  Z_{p,\lambda} =  \int_{\Rset} \left(  1 +  (1-\lambda) |x|^{p^*}
\right)^{\frac{1}{\lambda-1}}   dx  =   \frac{2  \,   B\left(   \frac{1}{p^*}  ,
    \frac{1}{1-\lambda}-\frac{1}{p^*} \right)}{p^* (1-\lambda)^{\frac{1}{p^*}}}$
the    normalization   coefficient    of   the    distribution   $g_{p,\lambda}$~\cite{Ber12:06_1,GraRyz07}).   Hence,  as  defined  in
Definition~(\ref{def:differential-escort_def}), $\opE_\alpha[g_{p,\lambda}](y) = \Big[
g_{p,\lambda}(x(y)) \Big]^\alpha$ with
\begin{eqnarray*}
\frac{dy}{dx} & = & \Big[ g_{p,\lambda}(x) \Big]^{1-\alpha}\\
& = & Z_{p,\lambda}^{\alpha-1} \left( 1 + (1-\lambda) \, |x|^{p^*}
\right)^{\frac{1-\alpha}{\lambda-1}}.
\end{eqnarray*}
Thus, $y(x)$ writes
\begin{eqnarray*}
y(x) & = & Z_{p,\lambda}^{\alpha-1} \, \textrm{sign}(x) \int_0^{|x|} \left( 1 +
(1-\lambda) \, t^{p^*} \right)^{\frac{1-\alpha}{\lambda-1}} dt\\
& = & \kappa_{p,\lambda,\alpha} \, \textrm{sign}(x) \int_0^{\frac{(1-\lambda)
|x|^{p^*}}{1 + (1-\lambda) |x|^{p^*}}} \tau^{\frac{1}{p^*}-1} \left( 1 - \tau
\right)^{\frac{\alpha-1}{\lambda-1}-\frac{1}{p^*} - 1} d\tau
\end{eqnarray*}
when  making the  change of  variables  $\tau =  \frac{(1-\lambda) t^{p^*}}{1  +
  (1-\lambda)  t^{p^*}}$   and  denoting  $\kappa_{p,\lambda,\alpha}   =  \frac{
  Z_{p,\lambda}^{\alpha-1}}{p^* (1-\lambda)^{\frac{1}{p *}}}$.
One can  recognize in the  integral the incomplete beta  function $\displaystyle
\Beta(a,b,x) =  \int_0^x t^{a-1} (1-t)^{b-1}  dt$ defined when $\Real\{a\}  > 0$
and  for  $x  \in  [0  \,  ;  \,  1)$~\cite{OlvLoz10}.   Here,  $a  =
\frac{1}{p^*} > 0$, $b  = \frac{\alpha-1}{\lambda-1} - \frac{1}{p^*}$ and noting
that $\frac{(1-\lambda)  |x|^{p^*}}{1 +  (1-\lambda) |x|^{p^*}} \in  [0 \,  ; \,
1)$. Hence,
\begin{equation}
y(x) = \kappa_{p,\lambda,\alpha} \, \textrm{sign}(x) \, \Beta\left( \frac{1}{p^*}
\, , \, \frac{\alpha-1}{\lambda-1} - \frac{1}{p^*} \, ; \, \frac{(1-\lambda)
|x|^{p^*}}{1+ (1-\lambda) |x|^{p^*}} \right).
\end{equation}
Note that \  ${\displaystyle \frac{y}{\kappa_{p,\lambda,\alpha}}}: \Rset \mapsto
\left(  -  B\left( \frac{1}{p^*}  ,  \frac{\alpha-1}{\lambda-1} -  \frac{1}{p^*}
  \right)  \,   ;  \,  B\left(  \frac{1}{p^*}   ,  \frac{\alpha-1}{\lambda-1}  -
    \frac{1}{p^*} \right) \right)$, where $B(a,b) = \lim_{x \to 1} \Beta(a,b,1)$
is the  beta function~\cite{OlvLoz10,  GraRyz07, AbrSte70}; $B(a,b)$ is thus
infinite when $b \le 0$.

Denoting $\Beta^{-1}$ the inverse of incomplete beta function, we obtain
\begin{equation}
1 + (1-\lambda) |x(y)|^{p^*} = \frac{1}{1-\Beta^{-1}\left( \frac{1}{p^*} \, , \,
\frac{\alpha-1}{\lambda-1} - \frac{1}{p^*} \, ; \,
\frac{|y|}{\kappa_{p,\lambda,\alpha}} \right)}
\end{equation}
and, thus,
\begin{equation}
\opE_\alpha\left[ g_{p,\lambda}\right](y) \propto \left[ 1-\Beta^{-1} \left(
\frac{1}{p^*} \, , \, \frac{\alpha-1}{\lambda-1} - \frac{1}{p^*} \, ; \,
\frac{|y|}{\kappa_{p,\lambda,\alpha}} \right) \right]^{\frac{\alpha}{1-\lambda}}
\un_{\left[ 0 \, ; \, B\left( \frac{1}{p^*} \, , \, \frac{\alpha-1}{\lambda-1} -
\frac{1}{p^*} \right) \right)} \left( \frac{|y|}{\kappa_{p,\lambda,\alpha}}
\right)
\label{eq:DifferentialScortDeformedGaussian}
\end{equation}
Note   that   from   $\Beta(a,-a,x)   =   a^{-1}   \left(\frac{x}{1-x}\right)^a$~\cite{GraRyz07,AbrSte70}),
we naturally recover that $\opE_1\left[ g_{p,\lambda} \right] = g_{p,\lambda}$.

Finally, let us remark that
\begin{equation}
\widetilde{\Dset}_p = \big\{ (\beta,\lambda) \in \Rset_+^{* \, 2}: \: 1-p^*\beta
< \lambda < 1 \big\} \: \cup \: \big\{ (\beta,\lambda) \in \Rset_+^{* \, 2}: \:
\lambda > 1 \big\} \: \cup \: \big\{ (\beta,1), \beta \in \Rset_+^* \big\},
\end{equation}
the first  ensemble being a subset  of $\opA[\Lset_p]$ and the  second one a
subset of $\opA[\overline{\Lset}_p]$. We treat now these three cases separately.


\subsubsection{The Case $1 - p^* \beta < \lambda < 1$}

Following  Appendix~\ref{app:StamDpTildaExistence},  let  us   first  determine
$\alpha$ such  that $\opA_\alpha(\beta,\lambda)  \in \Lset_p$, which is $\alpha$
such that $\alpha \beta = 1 + \alpha (\lambda-1)$. Hence,
\begin{equation}
\alpha = \frac{1}{\beta+1-\lambda} \qquad \mbox{and} \qquad
\opA_\alpha(\beta,\lambda) = \left( \frac{\beta}{\beta+1-\lambda} ,
\frac{\beta}{\beta+1-\lambda} \right).
\label{eq:AlphaLp}
\end{equation}
The fact  that $\beta > 0$ and  $\lambda < 1$ insures  that $\beta+1-\lambda \ne
0$.

From Sections~\ref{sec:Bercher} and~\ref{app:StamDpTildaExistence}, the minimizer of
the complexity is thus given by
\begin{equation}
\rho_{p,\beta,\lambda} = \opE_{\frac{1}{\beta+1-\lambda}} \left[
g_{p,\frac{\beta}{\beta+1-\lambda}} \right].
\label{eq:rho_Descort_lambdainf}
\end{equation}
One   can   easily    see   that   $\frac{\beta}{\beta+1-\lambda}   \in   \left(
  \frac{1}{1+p^*}  \,  ;   \,  1\right)$,  and  thus  we   immediately  get  from
Equation~\ref{eq:DifferentialScortDeformedGaussian},
\begin{equation}
\rho_{p,\beta,\lambda}(x) \propto \left[ 1-\Beta^{-1}\left(
\frac{1}{p^*} \, , \, \frac{\beta-\lambda}{1-\lambda} - \frac{1}{p^*} \, ; \,
\frac{|y|}{\kappa_{p,\alpha \beta,\alpha}} \right) \right]^{\frac{1}{1-\lambda}}
\un_{\left[ 0 \, ; \, B\left( \frac{1}{p^*} \, , \,
\frac{\beta-\lambda}{1-\lambda} - \frac{1}{p^*} \right) \right)}\left(
\frac{|y|}{\kappa_{p,\alpha \beta,\alpha}} \right).
\end{equation}
Noting  that   $\frac{\beta-\lambda}{1-\lambda}  =  \frac{\beta-1}{1-\lambda}  +
\frac{1}{p}$,  it   appears  that  this   density  is  nothing  more   than  the
$(p,\beta,\lambda)$-Gaussian                                                   of
Definition~\ref{def:triparametricGaussian} (remember  that  the families
  of density are defined up to a shift and a scaling).


\subsubsection{The Case $\lambda > 1$}

Following again Appendix~\ref{app:StamDpTildaExistence},  let us first determine
$\alpha$  such that  $\opA_\alpha(\beta,\lambda) \in  \overline{\Lset}_p$, i.e.,
such that $\alpha \beta = \frac{p^* +  1 - [ 1 + \alpha (\lambda-1)]}{p^*}$.  We
thus obtain
\begin{equation}
\alpha = \frac{p^*}{p^*\beta + \lambda - 1} \qquad \mbox{and} \qquad
\opA_\alpha(\beta,\lambda) = \left( \frac{p^* \beta}{p^*\beta + \lambda - 1} \,
, \, 1 + \frac{p^*(\lambda-1)}{p^*\beta + \lambda - 1} \right).
\label{eq:AlphaLpbar}
\end{equation}
The fact that $\beta > 0$ and $\lambda > 1$ insures that $p^*\beta + \lambda - 1
\ne 0$.

From   Section~\ref{sec:Bercher}    and   Appendix~\ref{app:StamDpTildaExistence}, the
minimizers for the complexity expresses
\begin{equation}
\rho_{p,\beta,\lambda} = \opE_{\frac{p^*}{p^*\beta+\lambda-1}} \left[
g_{p,1-\frac{p^*(\lambda-1)}{p^*\beta+\lambda-1}} \right].
\end{equation}
One  can  easily has  that  $1  - \frac{p^*(\lambda-1)}{p^*\beta+\lambda-1}  \in
\left(   1-p^*  \,   ;  \,   1\right)$  and   thus  we   immediately   get  from
Equation~\ref{eq:DifferentialScortDeformedGaussian}
\begin{equation}
\rho_{p,\beta,\lambda}(y) \propto \left[ 1-\Beta^{-1}\left( \frac{1}{p^*} \, ,
\, \frac{\beta-1}{\lambda-1} \, ; \,
\frac{|y|}{\kappa_{p,1-\alpha(\lambda-1),\alpha}} \right)
\right]^{\frac{1}{\lambda-1}} \un_{\left[ 0 \, ; \, B\left( \frac{1}{p^*} \, ,
\, \frac{\beta-1}{\lambda-1} \right) \right)}\left(
\frac{|y|}{\kappa_{p,1-\alpha(\lambda-1),\alpha}} \right).
\end{equation}
The  density is  again  nothing more  than  the $(p,\beta,\lambda)$-Gaussian  of
Definition~\ref{def:triparametricGaussian}.


\subsubsection{The Case $\lambda = 1$}

We  exclude  here  the  trivial  point  $\beta =  1$.   Now,  taking  $\alpha  =
\frac{1}{\beta}$  gives  $\opA_\alpha(\beta,1)=~(1,1)$.   We  know  that  the
minimizer for  $\beta =  1$ is given  by $g_{p,1}(x) =  Z_{p,1}^{-1} \exp\left(-
  |x|^{p^*} \right)$ with \linebreak $\displaystyle Z_{p,1} = \int_{\Rset} \exp(-|x|^{p^*})
dx  = \frac{2 \,  \Gamma\left(\frac{1}{p^*}\right)}{p^*}$~\cite{Ber12:06_1,GraRyz07}.

Following  again Appendix~\ref{app:StamDpTildaExistence},  we have  to determine
\begin{equation}
\opE_{\frac{1}{\beta}}\Big[        g_{p,1}       \Big](y)        =       \left[
  g_{p,1}(x(y))\right]^{\frac{1}{\beta}}       =      Z_{p,1}^{-\frac{1}{\beta}}
\exp\left(- \frac{|x(y)|^{p^*}}{\beta}\right)
\label{eq:rho_Descort_lambdaun}
\end{equation}
with
\begin{eqnarray*}
\frac{dy}{dx} & = & \Big[ g_{p,1} \Big]^{1-\frac{1}{\beta}}\\[4mm]
& = & Z_{p,1}^{\frac{1-\beta}{\beta}} \exp\left(- \frac{\beta-1}{\beta}
|x|^{p^*}\right),
\end{eqnarray*}
and thus
\begin{eqnarray*}
y(x) & = & Z_{p,1}^{\frac{1-\beta}{\beta}} \textrm{sign}(x) \int_0^{|x|} \exp\left(-
\frac{\beta-1}{\beta} t^{p^*}\right) dt.
\end{eqnarray*}
Viewing this integral in the complex plane (here in the real line), one can make
the change  of variables  $\tau = \frac{\beta-1}{\beta}  t^{p^*}$, i.e., \  $t =
\left( \frac{\beta-1}{\beta} \right)^{-\frac{1}{p^*}} \tau^{\frac{1}{p^*}}$ \ to
obtain
\begin{equation}
y(x) = \frac{Z_{p,1}^{\frac{1-\beta}{\beta}} }{p^* \left( \frac{\beta-1}{\beta}
\right)^{\frac{1}{p^*}}} \textrm{sign}(x) \int_0^{\frac{\beta-1}{\beta} |x|^{p^*}}
\tau^{\frac{1}{p^*}-1} \exp(-\tau) \, d\tau,
\end{equation}
where  $\left(  \frac{\beta-1}{\beta}  \right)^{\frac{1}{p^*}}$  is  complex  in
general, real only if $\frac{\beta-1}{\beta}  \ge 0$, i.e., if $\beta \not\in (0
\, ; \, 1)$.  \ One can  recognize in the integral the incomplete gamma function
$\displaystyle \G(a,x) = \int_0^x  t^{a-1} \exp(-t) dt$, defined for $\Real\{a\}
> 0$ and for any complex $x$~\cite{OlvLoz10}. We then obtain,
\begin{equation}
y(x) = \kappa_{p,\beta} \, \textrm{sign}(x) \, \left[ \left(
\frac{\beta-1}{\beta} \right)^{-\frac{1}{p^*}} \G\left(\frac{1}{p^*} \, ; \,
\frac{\beta-1}{\beta} \, |x|^{p^*} \right) \right],
\end{equation}
where  $\kappa_{p,\beta}  = \frac{Z_{p,1}^{\frac{1-\beta}{\beta}}}{p^*}$.   Note
that the term in square brackets is real and positive, and takes its values over
$\Rset_+$ if $\beta > 1$ (remember that we excluded the trivial situation $\beta
= 1$), and  over $\left[ 0 \, ;  \, \Gamma\left(\frac{1}{p^*}\right) \right)$ if
$\beta < 1$.

Denoting $\G^{-1}$ the inverse of the incomplete gamma function, this gives
\begin{equation}
\frac{1}{\beta} |x(y)|^p = \frac{1}{\beta-1} \G^{-1} \left( \frac{1}{p^*} \, ;
\, \left( \frac{\beta-1}{\beta} \right)^{\frac{1}{p^*}} \,
\frac{|y|}{\kappa_{p,1}}\right)
\end{equation}
defined for  $\frac{|y|}{\kappa_{p,1}} <  \frac{\Gamma(1/p^*)}{ \un_{(0 \,  ; \,
    1)}(\beta)}$ with the convention $1/0 = +\infty$. We thus achieve
\begin{equation}
\rho_{p,\beta,1}(y) \propto \exp\left( \frac{1}{1-\beta} \G^{-1} \left(
\frac{1}{p^*} \, ; \, \left( \frac{\beta-1}{\beta} \right)^{\frac{1}{p^*}} \,
\frac{|y|}{\kappa_{p,1}}\right) \right) \un_{\left[ 0 \, ; \,
\frac{\Gamma(1/p^*)}{\un_{(0;1)}(\beta)} \right)} \left(
\frac{|y|}{\kappa_{p,1}} \right).
\end{equation}
We again recover the $(p,\beta,\lambda)$-Gaussian.


\subsection{Symmetry through the Involution $\opT_p$.}

For $\lambda = 1,$ the result is trivial since $\opT_p(\beta,1) = (\beta,1)$ (see
Equation~\ref{eq:Involution}).

Now, for $\lambda \ne 1$, let us denote $(\overline{\beta},\overline{\lambda}) =
\opT_p(\beta,\lambda)  =  \left( \frac{p^*  \beta  +  \lambda-1}{p^* \lambda}  ,
  \frac{1}{\lambda}\right)$ the involutary transform of $(\beta,\lambda)$.  Some
simple  algebra allows to  show that  if $1  - \beta  p^* <  \lambda <  1$, then
$\overline{\lambda} > 1$, and reciprocally.   Thus, it is straightforward to see
that     $q_{p,\opT(\beta,\lambda)}    =    q_{p,\beta,\lambda}$     and    that
$\frac{1}{|1-\overline{\lambda}|} = \frac{\lambda}{|1-\lambda|}$, leading to
\begin{equation}
  g_{p,\opT_p(\beta,\lambda)} \propto \Big[ g_{p,\beta,\lambda} \Big]^\lambda.
\end{equation}

Now, if  $\lambda <  1$, the  optimal bound is  given by  $K_{p,\beta,\lambda} =
\alpha^2  \,  K_{p,\alpha \beta  ,  \alpha \beta}$  (see~Equations~\ref{eq:AlphaLp}
and~\ref{eq:AffineComplexity}).   Then,  $\overline{\lambda}  >  1$  and  thus
$K_{p,\opT_p(\beta,\lambda)}  =  \overline{\alpha}^2  \,  K_{p,\overline{\alpha}
  \overline{\beta}    ,   1    +    \overline{\alpha}   (\overline{\lambda}-1)}$
(see~Equations~\ref{eq:AlphaLpbar}  and~\ref{eq:AffineComplexity},  where $\alpha$
is  here  denoted  by  $\overline{\alpha}$ and  $(\beta,\lambda)$  is  obviously
replaced   by    $(\overline{\beta},\overline{\lambda})$).    Simple   algebraic
manipulations allow  us to see that  $\overline{\alpha} = \frac{\lambda}{\beta}$
and that  $\opT_p(\alpha\beta,\alpha\beta) = (\overline{\alpha} \overline{\beta}
,      1      +      \overline{\alpha}      (\overline{\lambda}-1))$,      hence
$K_{p,\opT_p(\beta,\lambda)}   =  \left(   \frac{\lambda}{\beta}   \right)^2  \,
K_{p,\opT_p(\alpha\beta,\alpha\beta)}      =      (\lambda     \alpha)^2      \,
K_{p,\alpha\beta,\alpha\beta}$ from  Proposition~\ref{prop:Symmetries}.  We then
obtain again  $K_{p,\opT_p(\beta,\lambda)} = \lambda^2  \, K_{p,\beta,\lambda}$.
The  case $\lambda  >  1$ is  treated  in a  similar way,  leading  to the  same
conclusion.


\subsection{Explicit Expression of the Lower Bound.}

Let us first  consider the case $\lambda <  1$. Thus, $\zeta_{p,\beta,\lambda} =
\beta$               (see               Equation~\ref{eq:zeta}).               From
Equations~\ref{eq:AlphaLp}and~\ref{eq:rho_Descort_lambdainf}                      and
Equation~\ref{eq:AffineComplexity}, we have
\begin{eqnarray*}
K_{p,\beta,\lambda} & = & \alpha^2 \, K_{p,\alpha \beta,\alpha \beta}\\[4mm]
& = & \frac{(\alpha \beta)^2 \, K_{p,\alpha \beta,\alpha \beta}}{\beta^2}
\end{eqnarray*}
that        is,         noting        that        $\alpha         \beta        =
\frac{\zeta_{p,\beta,\lambda}}{\zeta_{p,\beta,\lambda} + |1 - \lambda|}$,
\begin{equation}
K_{p,\beta,\lambda} = \frac{\left(
\frac{\zeta_{p,\beta,\lambda}}{\zeta_{p,\beta,\lambda} + |1-\lambda|} \right)^2
\, K_{p,\frac{\zeta_{p,\beta,\lambda}}{\zeta_{p,\beta,\lambda} + |1-\lambda|} ,
\frac{\zeta_{p,\beta,\lambda}}{\zeta_{p,\beta,\lambda} +
|1-\lambda|}}}{\zeta_{p,\beta,\lambda}^2},
\label{eq:Kopt_zeta}
\end{equation}
when   $\lambda    >   1$.     Thus,   $\zeta_{p,\beta,\lambda}   =    \beta   +
\frac{\lambda-1}{p^*}$ (see Equation~\ref{eq:zeta}).  Denoting $(\overline{\beta} ,
\overline{\lambda}) = \opT_p(\beta,\lambda)$ (see Equation~\ref{eq:Involution}) and
noting  that   $\overline{\lambda}  =   \frac{1}{\lambda}  <  1$   and  applying
successively     Equation~\ref{eq:KInvolution}    (see     previous    subsection),
Equations~\ref{eq:AlphaLp},           \ref{eq:rho_Descort_lambdainf}           and~\ref{eq:AffineComplexity}     (where    $\alpha$    is     denoted    here
$\overline{\alpha}$   and    $(\beta,\lambda)$   is   obviously    replaced   by
$(\overline{\beta} , \overline{\lambda})$), we have
\begin{eqnarray*}
K_{p,\beta,\lambda}
& = & \frac{1}{\lambda^2} \, K_{p,\overline{\beta},\overline{\lambda}}\\[4mm]
& = & \frac{(\overline{\alpha} \overline{\beta})^2 \, K_{p,\overline{\alpha}
\overline{\beta} , \overline{\alpha} \overline{\beta}}}{\lambda^2
\overline{\beta}^2}.
\end{eqnarray*}
It  is straightforward  to  see  that $\lambda^2  \overline{\beta}^2  = \beta  +
\frac{\lambda-1}{p^*}  =  \zeta_{p,\beta,\lambda}$  and that  $\overline{\alpha}
\overline{\beta}  = \frac{p^*  \beta  + \lambda-1}{p^*  \beta  + \lambda  - 1  +
  p^*(\lambda-1)}   =  \frac{\zeta_{p,\beta,\lambda}}{\zeta_{p,\beta,\lambda}  +
  |\lambda-1|}$ so that Equation~\ref{eq:Kopt_zeta} still holds.

The  case   $\lambda  =  1$  can  be   viewed  as  the  limit   case,  or  using
Equations~\ref{eq:rho_Descort_lambdaun} and \ref{eq:AffineComplexity}   to  conclude
that Equation~\ref{eq:Kopt_zeta} still holds. It remains to evaluate $l^2 K_{p,l,l}
=   l^2    C_{p,l,l}(g_{p,l})$   with   $l   \le   1$.     The   evaluation   of
$\sqrt{N_l(g_{p,l})}$ and  $\sqrt{F_{p,l}(g_{p,l})}$ was conducted  for instance
in~\cite{LutYan05}, which gives with our notations, for $l < 1$
\begin{equation}
l^2 \, K_{p,l,l} = \left[ \frac{2}{p^*} \left( \frac{p^* l}{1-l}
\right)^{\frac{1}{p^*}} \left( \frac{p^* l}{(p^*+1) l - 1}
\right)^{\frac{l}{1-l}+\frac{1}{p}} B\left( \frac{1}{p^*} , \frac{1}{1-l} -
\frac{1}{p^*} \right) \right]^2
\end{equation}
and
\begin{equation}
K_{p,1,1} = \left[ \frac{2 \, e^{\frac{1}{p^*}} \Gamma\left( \frac{1}{p^*}
\right)}{p^{* \, {\frac{1}{p}}}} \right]^2.
\end{equation}
Noting  that  $\frac{1}{1-l}-\frac{1}{p^*} =  \frac{l}{1-l}  + \frac{1}{p}$  and
taking    $l    =    \frac{\zeta_{p,\beta,\lambda}}{\zeta_{p,\beta,\lambda}    +
  |1-\lambda|}$, we achieve the wanted result from Equation~\ref{eq:Kopt_zeta}.



\end{document}